\documentclass[11pt]{article}

\usepackage{amsfonts}
\usepackage{amsmath}
\usepackage{amssymb}
\usepackage{amsthm}
\usepackage{mathrsfs}
\usepackage{graphicx}
\usepackage{xcolor}

\usepackage[]{algorithm2e}

\theoremstyle{plain}
\newtheorem{theorem}{Theorem}
\newtheorem{proposition}[theorem]{Proposition}
\newtheorem{definition}[theorem]{Definition}

\theoremstyle{definition}
\newtheorem{example}[theorem]{Example}

\numberwithin{exercise}{section}
\numberwithin{equation}{section}
\numberwithin{theorem}{section}
\numberwithin{problem}{section}

\numberwithin{figure}{section}

\DeclareMathOperator{\diag}{diag}

\newcommand{\bs}[1]{{\boldsymbol{#1}}}

\newcommand{\R}{\mathbf{R}}

\newcommand{\D}{\,\mathrm{d}}

\newcommand{\IP}[2]{\left\langle#1\,,#2\right\rangle}

\begin{document}

\title{A mathematical framework of consumer--resource dynamics: How to incorporate interactions between interactions in evolutionary process}

\author{Alexander S. Bratus$^{1,2,3,}$\footnote{e-mail: alexander.bratus@yandex.ru}$\,\,$, Sergei V. Drozhzhin$^{2,3},$ Artem S. Novozhilov$^{4,}$\footnote{e-mail: artem.novozhilov@ndus.edu} \\[3mm]
\textit{\normalsize $^\textrm{\emph{1}}$Marchuk Institute of Numerical Mathematics, Russian Academy of Sciences, 119333, Russia}\\[0mm]
\textit{\normalsize $^\textrm{\emph{2}}$Moscow Center of Fundamental and Applied Mathematics, 119991, Russia}\\[0mm]
\textit{\normalsize $^\textrm{\emph{3}}$Institute of Management and Digital Technologies, Russian University of Transport, 127055, Russia}\\[0mm]
\textit{\normalsize $^\textrm{\emph{4}}$Department of Mathematics, North Dakota State University, Fargo, ND, 58108, USA}}

\date{}

\maketitle

\begin{abstract}A novel mathematical framework is proposed to describe the ecological and evolutionary consequences of consumer--resource interactions. Both the consumer and resource are assumed to consist of several (sub)species, which interact between themselves in addition to incorporating the deleterious effects of the consumer on the resource. Separating the ecological and evolutionary time scales, we allow our mathematical model to evolve, with the evolutionary steps chosen according to the (divergent) objective functions of the consumer and the resource. Numerical simulations show that the model, along with the expected outcomes of either consumer or resource winning the evolutionary battle, is capable of producing also the (quasi)stationary state of consumer--resource coexistence with monotone growth of both the consumer and resource fitnesses. Such stable states highlight the importance of the intra-population interactions, which, despite the opposite evolutionary goals of the consumer and the resource, lead to long term ecological stability.



\end{abstract}

\section{Introduction}Consumer--resource interactions have been in the center of evolutionary research for many decades. Mathematical modeling provides an invaluable tool to study various aspects of such interactions, ranging from the basic and now classical Lotka--Volterra predator--prey model to intricate immune system--virus individual based models. An important, and often overlooked thing is that in many situations, both consumer and resource are not unique entities; they are frequently represented by closely related communities of interacting subpopulations, whose structure is shaped by their own connections.

To give one specific example of such interaction, consider the general process of virus infection \cite{payne2022viruses}. After the virus has been attached to a cell and the process of delivery of the virus genome to the cell's cytoplasm has started, it is possible, in general, to have different outcomes of such interaction between the host cell and the virus. The cell--virus interaction may lead to the death of the cell. Or the defence mechanisms of the cell can eventually destroy the virus population. Finally, it is also possible for the virus and cell to coexist in a stable stationary state, at which the infected cells synthesize and release viral particles, continuing at the same time to perform the normal activities necessary for the organism. Two important aspects of the latter outcome are that the cell death is not a necessary condition to release the virus, and the spread of certain viruses can occur by sequential infections of one cell after another. In this case the virus can effectively escape the antibodies and thus can be preserved in the host body for extended periods of time  \cite{stebbing2003virus,wilson2001viral}. Thus in general we have two communities (the cell community and the virus community) which exhibit both intra- and inter- species interactions. From the modeling perspective these interactions are often studied on their own, and yet it is clear that these ``interactions between interactions'' can have a dramatic effect on each other \cite{moller2008interactions}, producing inevitable ecological and evolutionary consequences.

The main goal in this paper is to present a novel mathematical framework to describe such consumer--resource systems, including both internal and external interactions in both communities and eventually allowing to incorporate evolutionary adaptations. As an abstract mathematical description of the host population we consider the so-called open replicator systems \cite{bratus2009stability,pavlovich2012studying,yegorov2020open} (for the precise detail see below), which are based on now classical replicator equation \cite{cressman2014replicator,hofbauer2003egd}
\begin{equation}\label{eq1:1}
  \dot w_i = w_i\Bigl((\bs A \bs w)_i-\IP{\bs{Aw}}{\bs w}\Bigr),\quad i=1,\ldots,n,
\end{equation}
where $\bs w(t)=(w_1(t),\ldots,w_n(t))^\top$ is the vector of (in one possible interpretation) interacting macromolecules at time $t$, the matrix $\bs A=[a_{ij}]$ describes the rates with which the macromolecules of type $j$ catalyze the replication of macromolecules of type $i$, $(\bs A \bs w)_i$ is the $i$-th element of the vector, and $\IP{\cdot}{\cdot}$ is the standard dot product in $\R^n$. As it is straightforward to check, the equation \eqref{eq1:1} is written in the way to keep the simplex $S_n=\{\sum_{i=1}^n w_i=1,\,w_i\geq 0\}$ invariant, and thus we interpret $w_i(t)$ as \textit{relative} concentrations. In such formulation, ecologically, it is impossible to consider the extinction of the community, which, as noted above, is an important outcome. Therefore, as a first step, we introduce the open replicator systems, which are described in Section \ref{sec:2}. These systems inherent the rich dynamical regimes of the classical equation \eqref{eq1:1}, but also naturally allow for all the macromolecule concentrations to be attracted to the origin, which corresponds to the total population collapse. We illustrate the general theory with one specific replicator system -- the hypercycle equation -- whose dynamics is well understood \cite{hofbauer1998ega}. In particular, we show by means of numerical calculations that in the phase space of our system there exists a compact ball, which is positively invariant. We call such systems conditionally permanent.

In Section \ref{sec:3} we couple the dynamics of the open replicator system with another replicator, which plays the role of the consumer (parasite or virus) and relies on the presence of different host (resource) macromolecules to survive. An important property of the interaction between the cell population and virus population is the appearance of the system dynamics with opposite objective functions. While it is certainly true that the real life is more complicated \cite{birch2016natural,bratus2018adaptive, bratus2020geometry} than our description, as a first natural approximation we assume that the ultimate goal of each interacting community is to increase its own average fitness. To accommodate this assumption we split the evolutionary process into consecutive steps. Building on the existing approach to evolutionary adaptation of different replicator systems \cite{bratus2024food,bratus2018evolution,drozhzhin2021fitness}, we explicitly include this evolutionary process with two opposite objective functions in our model as a series of linear programming problems. At each odd step of this process an objective function is chosen to satisfy the evolutionary goals of the cell, and at each even step, it is virus' turn to maximize its average fitness based on the current average fitness of the replicator. This algorithm is described in detail in Section \ref{sec:4}. In Section \ref{sec:5} we demonstrate, by means of numerical examples, that the proposed mathematical framework naturally leads to various evolutionary outcomes. We conclude our paper with discussion of the main findings and possible future directions.

\section{Open replicator systems}\label{sec:2}The goal for this section is to define a system of ordinary differential equations, which inherits a number of important properties of the classical replicator equation \eqref{eq1:1}, and yet allows a possibility for the solutions to be attracted to the origin, which biologically means the total population extinction. Mathematically, it amounts to changing the state space of the system from the simplex to the non-negative cone $\R^n_{+}$ of $\R^n$.

Consider a population of $n$ types of macromolecules. Let $w_i(t)$ be the population size of the $i$-th type at time $t\geq 0$, $i=1,\ldots, n$, and $\bs w(t)=(w_1(t),\ldots,w_n(t))^\top\in\R^n_+$. The interactions (the reaction rates) are defined, as usual, by the elements of real $n\times n$ matrix $\bs A=[a_{ij}]$, where the element $a_{ij}$ describes how type $j$ macromolecule helps the replication of type $i$ macromolecule. The system we consider in the following takes the form
\begin{equation}\label{eq2:1}
    \dot w_i=w_i\Bigl(\psi(\bs w)(\bs{Aw})_i-d_i\Bigr),\quad i=1,\ldots, n.
\end{equation}
Here $\psi(\bs w)=\exp(-\gamma S(\bs w)),\,S(\bs w)=\sum_{i=1}^n w_i,\,\gamma>0$, $(\bs{Aw})_i$ is the $i$-th element of the vector $\bs{Aw}$, and $d_i>0$ are constants that can be interpreted as dissipation rates. Given the initial conditions $w_i(t_0)=w_i^0\geq 0$, we immediately conclude that system \eqref{eq2:1} is invariant in $\R^n_+$. The key difference of system \eqref{eq2:1} from the classical replicator equation \eqref{eq1:1} is that it does not confine the system dynamics to the unit simplex in $\R^n_+$ and allows the simultaneous extinction of all types of macromolecules, $\bs w(t)\to 0$ as $t\to \infty$ since, clearly, $0$ is an asymptotically stable equilibrium point of \eqref{eq2:1}. A similar approach to the quasispecies equation was considered previously in \cite{yegorov2020open}, and a special from of equations \eqref{eq2:1} was analyzed in \cite{pavlovich2012studying}.

To interpret our model biologically we must confirm that the total population size $S(\bs w)$ does not grow without bounds. The following proposition holds, and the proof in given in Appendix \ref{ap:1}.
\begin{proposition}\label{pr:1}The solutions to \eqref{eq2:1} are unique, exit for all $t\geq 0$, and bounded if $\min_i\{d_i\}>0$.
\end{proposition}

In many cases we are interested in replicator systems, for which \textit{all} the macromolecules survive in the long run. The technical definition in this case is that the system is \textit{permanent} or \textit{uniformly persistent} (see, e.g., \cite{smith2011dynamical} for a careful treatment). In our case \eqref{eq2:1}, due to the attracting nature of the origin, we cannot call this system permanent, using the definition from \cite{smith2011dynamical}. As examples show, however, it is still possible to have multiple instances of dynamics, in which, choosing the right initial conditions, we can expect that all the concentrations are contained in a compact ball in $\R^n_+$ and thus separated from the extinction boundary of $\R^n_+$. Only such systems are considered in the following. Since it is important for the presentation, we give a formal
\begin{definition}An open replicator system \eqref{eq2:1} is called conditionally permanent if there is a non-empty compact set $B$ strictly inside $\R^n_+$, such that if $\bs w(0)\in B$ then $\bs w(t)\in B$ for all $t>0$.
\end{definition}

A necessary condition for \eqref{eq2:1} to be conditionally permanent is that our system must possess a strictly positive equilibrium point $\bs{\hat w}>0$, which means that all $\hat w_i>0$.

To find the equilibria \eqref{eq2:1} we assume for the following that $\bs A$ is invertible. If we want an equilibrium $\bs{\hat w}$ with all the coordinates positive, it must be true that, denoting $\hat S=\sum_{i=1}^n\hat w_i$,
$$
(\bs{A\hat{w}})_i e^{-\gamma \hat S}=d_i,
$$
or, using $\bs d=(d_1,\ldots,d_n)$, in the matrix form
$$
\hat w_i=(\bs A^{-1}\bs d)_ie^{\gamma \hat S}.
$$
Summing all the equations together, we find
$$
\hat S e^{-\gamma\hat S}=\sum_{i=1}^n(\bs{A}^{-1}\bs d)_i=\IP{
\bs{A}^{-1}\bs d}{\bs 1}=a,\quad \bs 1=(1,\ldots, 1)^\top\in\R^n.
$$
Hence, using the explicit form of $\psi(S)$, we conclude that if $a>(\gamma e)^{-1}$ then there are no equilibria with all the coordinates different from zero, and if
\begin{equation}\label{eq2:1a}
    a=\IP{\bs{A}^{-1}\bs d}{\bs 1}\leq (\gamma e)^{-1},
\end{equation}
then it is possible to have a positive equilibrium. More specifically, if $a=(\gamma e)^{-1}$ then such equilibrium is unique, and if $a<(\gamma e)^{-1}$ then it is possible to have two such equilibria $\bs{\hat{w}}^{1,2}$ for which $S(\bs{\hat{w}}^1)<\gamma^{-1}$ and $S(\bs{\hat{w}}^2)>\gamma^{-1}$ respectively. We cannot say much about the properties of these equilibria in full generality, and yet if the system is conditionally permanent then at least one such equilibrium must exist (note that we do not claim anything about its stability). To be precise, we can prove (for the proof see Appendix \ref{ap:1})

\begin{proposition}\label{pr:2}Assume that in system \eqref{eq2:1} it holds that $\det \bs A\neq 0$ and the system is conditionally permanent. Then there exists an equilibrium point $\bs{\hat{w}}\in\R^n_+$ such that $\bs{\hat w}>0$ with coordinates satisfying
\begin{equation}\label{eq2:2}
    \hat w_i=(\bs A^{-1}\bs d)_ie^{\gamma \hat S},\quad \bs d=(d_1,\ldots,d_n)^\top,\quad \hat S=\sum_{i=1}^n\hat w_i.
\end{equation}
\end{proposition}

In general system \eqref{eq2:1} may also have equilibria with $k$ zero coordinates and $n-k$ non-zero coordinates; in this case, a condition for existence of equilibria, similar to \eqref{eq2:1a}, still holds with minor modifications. For instance, if the we assume that we are looking at equilibria with the first $s$ coordinates zero, and $n-s$ nonzero coordinates let us call it $\bs{\hat{w}}^s$. Denoting $\bs 1^s$ the vector with $s$ first zeroes and $n-s$ ones, similar to the above we find that
\begin{equation}\label{eq2:1b}
    \IP{\bs{A}^{-1}\bs d}{\bs 1^s}\leq (\gamma e)^{-1}.
\end{equation}

As a specific example of the general theory above consider the classical hypercycle equation \cite{Eigen1977,hofbauer1998ega} with the matrix
\begin{equation}\label{eq2:1c}
    \bs A=\begin{bmatrix}
            0 & 0 & \ldots & 0 & 1 \\
            1 & 0 & \ldots & 0 & 0 \\
            \vdots &  &  &  &  \\
            0 &  & \ldots & 1 & 0 \\
          \end{bmatrix}.
\end{equation}
\begin{proposition}\label{pr:3} Consider the open replicator equation \eqref{eq2:1} with matrix \eqref{eq2:1c}. Assume that $d_1=\ldots=d_n=d$ and assume that strict inequality holds in \eqref{eq2:1a}. Then defined above equilibrium $\bs{\hat{w}}^1$ is unstable for all $n$, equilibrium $\bs{\hat{w}}^2$ is asymptotically stable for $n=3,4$ and becomes unstable for $n\geq 5$.
\end{proposition}
The idea of the proof is based on the explicit calculation of the Jacobi matrix at the corresponding equilibria, more details are given in Appendix \ref{ap:1}. Note that in the general case $d_i\neq d$ invoking the continuous dependence of the eigenvalues on the elements of the matrix, Proposition \ref{pr:3} can be generalized to a small enough neighborhood of vector $(d,\ldots,d)^\top$. If the condition \eqref{eq2:1a} does not hold, then it can be shown that the origin is globally asymptotically stable (\cite{pavlovich2012studying}).

We do not have analytical proofs for the following facts, but extensive numerical computations suggest that for the open hypercycle system \eqref{eq2:1}, \eqref{eq2:1c} in the case when strict inequality \eqref{eq2:1a} holds, the unstable manifold of $\bs{\hat{w}}^1$ divides the positive cone $\R^n_+$ into two parts $B_1$ and $B_2$. If the initial conditions belong to $B_1$ then asymptotically all solutions tend to the origin, if, however, the initial conditions are chosen from $B_2$, then the system becomes conditionally permanent, i.e., the solutions never approach zero and stay inside $B_2$, which, by Proposition \eqref{pr:2} must contain the equilibrium $\bs{\hat{w}}^2$.

\begin{example}\label{ex2:5}
As a specific example consider the case $n=5$. It is well known that in this case the classical hypercyclic equation is permanent, and, moreover, the internal positive equilibrium is unstable, and there exists an asymptotically stable limit cycle \cite{hofbauer1998ega}. For the open replicator equation we choose $\bs d=(0.0648, 0.0577, 0.0444, 0.439, 0.0468)$ and $\gamma=0.1$. Numerically we observe that at least for some initial conditions the system settles on an asymptotically stable limit cycle (see Fig. \ref{fig2:1}).

\begin{figure}
\centering
\includegraphics[width=0.6\textwidth]{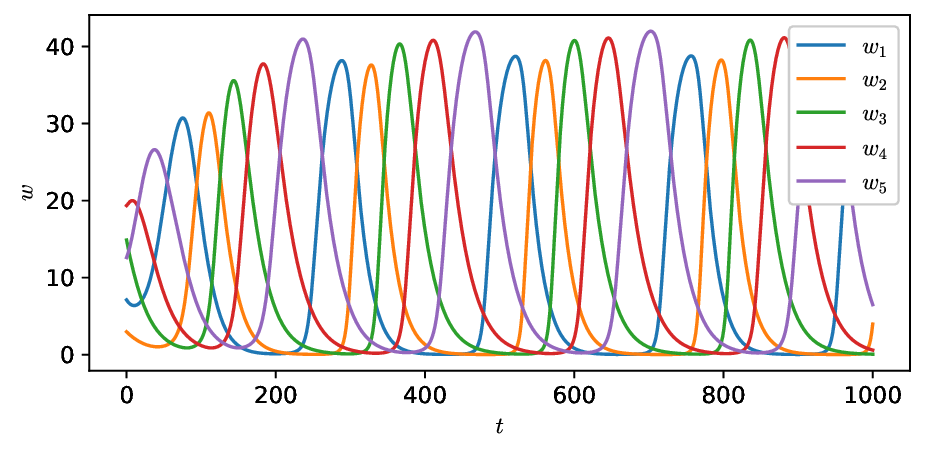}
\caption{Limit cycle in open replicator system with $n=5$. For the exact parameter values see the text.}\label{fig2:1}
\end{figure}

A more careful numerical analysis of this system suggests that the system is conditionally permanent. According to the general theory outlines above, for the chosen parameter values there exist two positive equilibria $\bs{\hat{w}}^1$ and $\bs{\hat{w}}^2$, which are both unstable. However, the unstable manifold of $\bs{\hat w}^1$ divides the phase space into the two connected parts: if the initial conditions belong to the first one then the system dynamics leads to the population extinction ($\bs w(t)\to 0$ as $t\to\infty$); if the initial conditions are picked from the second part, all the concentrations $\bs{w}(t)$ are separated from zero for all future $t$ (see Fig. \ref{fig2:2}, where the basin of attraction of the origin is shown). Thus it is highly probably, based on the observed numerical experiments, that the system is conditionally permanent in this case.
\begin{figure}[!th]
\centering
\includegraphics[width=0.495\textwidth]{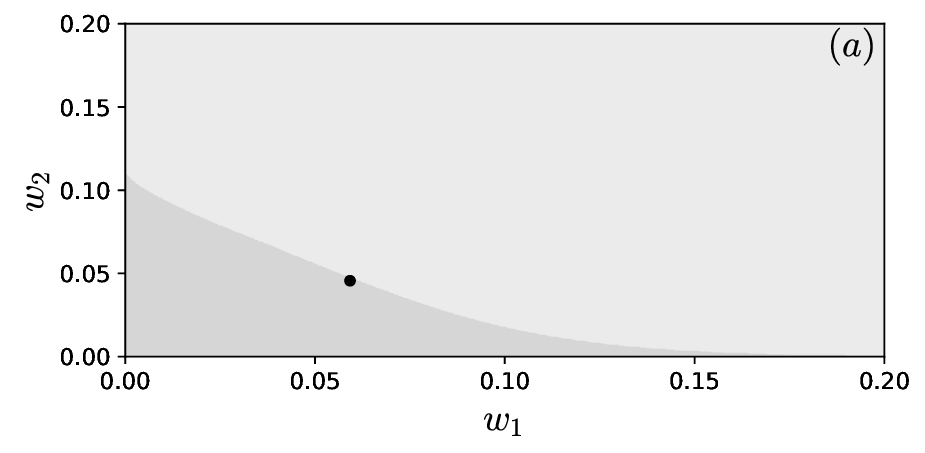}
\includegraphics[width=0.495\textwidth]{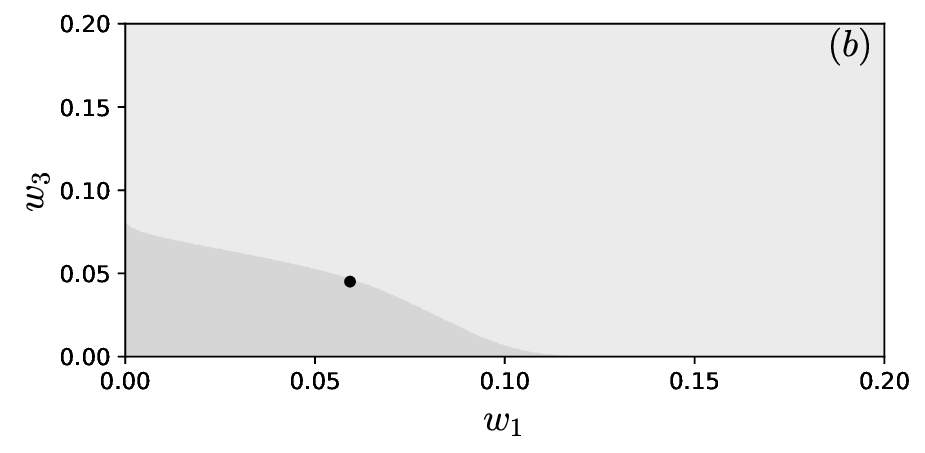}
\includegraphics[width=0.495\textwidth]{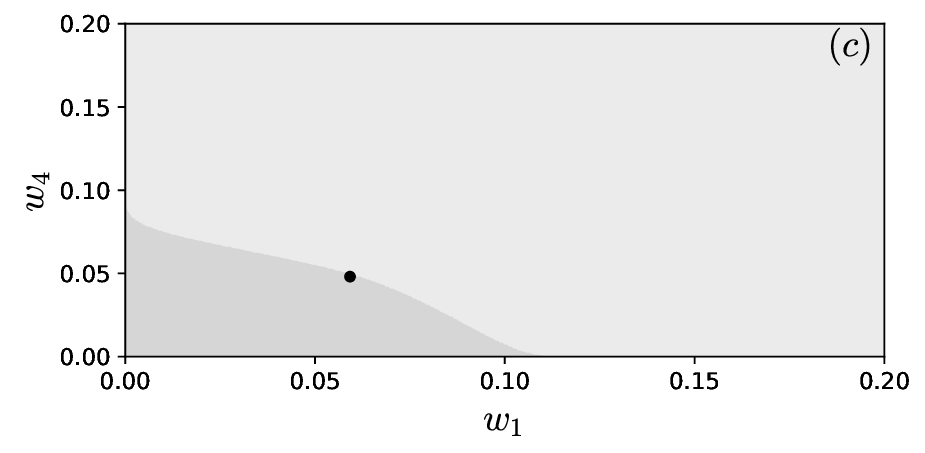}
\includegraphics[width=0.495\textwidth]{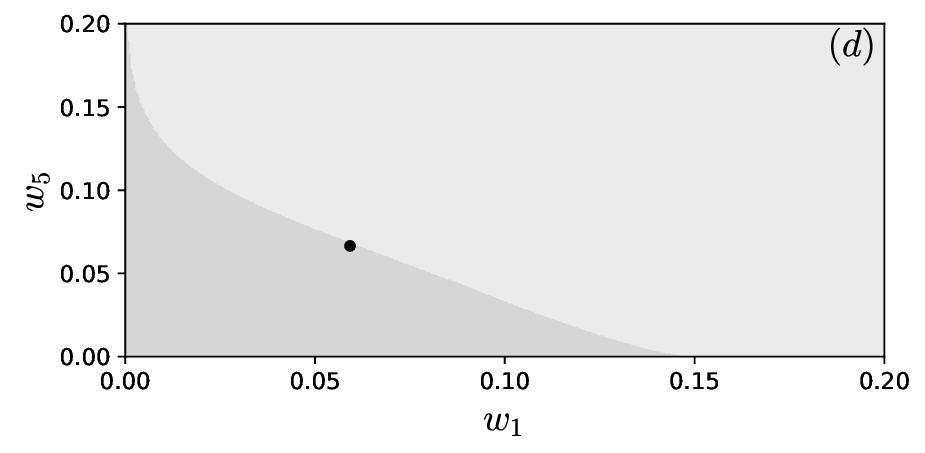}
\caption{The basin of attraction of the origin in the open hypercyclic equation (Example \ref{ex2:5}). The dot in the figures is the equilibrium $\bs{\hat w}^1$, whose unstable manifold divides the phase space into two parts. If the initial conditions belong to the darker regions then $\bs{w}(t)\to 0$. Otherwise $\bs{w}(t)$ is separated from zero (see Fig. \ref{fig2:1} for an example). The panels represent the cross sections of the phase space from which the initial conditions (the rest of of the initial conditions are taken to be equal to the coordinates of $\bs{\hat w}^1$). }\label{fig2:2}
\end{figure}
\end{example}

In addition to the details of the dynamical behavior discussed above, the coordinates of the positive equilibrium $\bs{\hat{w}}^2$, even in the case when this equilibrium is unstable, are still carrying a significant information about the averages of the solutions along the orbits (see the proof of Proposition \ref{pr:2}). This will allow us to consider the process of \textit{evolutionary adaptation} for the open replicator systems \eqref{eq2:1} (see \cite{bratus2024food,bratus2018evolution} for more details). In short, assuming the time scale separation, we can allow for the elements of matrix $\bs A$ to evolve according to a chosen criterion. This evolution will be observed as long as the evolved close system remains conditionally permanent, for which the necessary condition is to have the positive equilibrium $\bs{\hat{w}}^2$.

To specify the optimization criterion that will guide the system evolution, we define, following \cite{yegorov2020open}, the \textit{mean fitness at the equilibrium} $\bs{\hat{w}}$ as
\begin{equation}\label{eq2:3}
 \bar{f}(\bs{\hat w})=\frac{\IP{\bs{A\hat{w}}}{\bs{\hat w}}}{\IP{\bs D \bs{\hat w}}{\bs 1}}\,,\quad \bs D=\diag(d_1,\ldots,d_n).
\end{equation}
To motivate the expression in \eqref{eq2:3} we note that the numerator gives the usual mean fitness of the classical replicator equation \eqref{eq1:1}, whereas the denominator describes the damage to the system from the degradation rates; from an idealistic point of view (see, e.g., \cite{bratus2018adaptive} for a thorough discussion and references therein) the system should evolve in a way to maximize the numerator (within the allowed constrains) and minimize the denominator.

Since by summing all the equations in \eqref{eq2:1} we get
$$
\psi(\bs{\hat w})\IP{\bs{A\hat{w}}}{\bs{\hat w}}=\IP{\bs D\bs{\hat w}}{\bs 1},
$$
then we conclude that the average fitness of the open replicator equation \eqref{eq2:1} is given by
\begin{equation}\label{eq2:4}
    \bar f(\bs{\hat w})=\psi^{-1}(\bs{\hat w})=\exp(\gamma \hat S)=\exp(\gamma S(\bs{\hat w})).
\end{equation}
We expect that the adaptive evolution of the open replicator systems will be similar to the one observed in the classical replicator equation (see \cite{bratus2018evolution,drozhzhin2021fitness}). To make a next step in our analysis, we couple the replicator dynamics with another system, whose ultimate goal is to use the existing system as a resource to guarantee its own existence. We dub this second system as ``virus infection'' keeping in mind a possible biological interpretation.

\section{A mathematical model of interaction of open replicator system and virus infection}\label{sec:3}
Consider the following $n+k$ dimensional system of ordinary differential equations:
\begin{equation}\label{eq3:1}
\begin{split}
    \dot w_i&=w_i\psi(\bs w)\Bigl((\bs{Aw})_i-\alpha (\bs B \bs v)_i\Bigr)-d_iw_i,\quad w_i(0)>0,\quad i=1,\ldots, n,\\
    \dot v_j&=v_j\Bigl(-\beta v_j+(\bs B^\top \bs w)_j\Bigr),\quad v_j(0)>0,\quad j=1,\ldots,k,\quad k<n.
\end{split}
\end{equation}
Here $\bs w(t),\bs A,\,\bs d$ as before, $v_j(t)$ is the concentration of viral types of macromolecules, which we call ``virus infection,'' $j=1,\ldots, k $, $\bs v(t)=(v_1(t),\ldots,v_k(t))^\top$, $\alpha, \beta$ are positive parameters, and matrix $\bs B=[b_{ij}]$ is an $n\times k$ matrix with nonnegative elements, which describe the deleterious effect of the virus molecule of type $j$ on the replicator molecule of type~$i$. We note that since $\IP{\bs B^\top\bs w}{\bs v}=\IP{\bs w}{\bs{Bv}}\geq 0$, then Proposition \ref{pr:1} can be immediately generalized to system \eqref{eq3:1}, i.e, the solutions to this system starting in the nonnegative cone $\R^{n}_+\times \R^k_+$ exist, unique, and uniformly bounded, hence the system biologically meaningful.

Similar to the discussion above, we are looking primarily to the positive equilibria of \eqref{eq3:1}. Equilibria $(\bs{\hat w},\bs{\hat v})$ of \eqref{eq3:1} are given by
\begin{equation}\label{eq3:2}
    \begin{split}
    \hat w_i\psi(\bs{\hat w})\Bigl((\bs A\bs{\hat w})_i-\alpha(\bs B\bs{\hat v})_i\Bigr)-d_i\hat w_i&=0,\quad i=1,\ldots,n,\\
    \hat v_j\Bigl(-\beta \hat v_j+(\bs B^\top \bs{\hat w})_j\Bigr)&=0,\quad j=1,\ldots, k.
    \end{split}
\end{equation}
Clearly if $\bs{\hat w}=0$ then $\bs{\hat v}=0$ and the origin is asymptotically stable in $\R^n\times\R^k$. If $\bs{\hat v}=0$ then we restore the equilibria of the open replicator system \eqref{eq2:1}.

For the following analysis assume that $\bs{\hat v}\neq 0$, then from the second equation in \eqref{eq3:2} it follows
\begin{equation}\label{eq3:3}
    \bs{\hat v}=\beta^{-1}\bs B^\top\bs{\hat w}.
\end{equation}
Denoting $\bs C=\bs B\bs B^\top$ and $\sigma = \alpha/\beta$ and assuming that $\bs{\hat w}\neq 0$, \eqref{eq3:1} and \eqref{eq3:2} imply that
\begin{equation}\label{eq3:4}
    \psi(\bs{\hat w})((\bs A-\sigma \bs C)\bs{\hat w})_i=d_i,\quad i=1,\ldots,n.
\end{equation}
If $\bs A-\sigma\bs C$ is invertible then, summing the equations in \eqref{eq3:4}, we find that
$$
\psi(\bs{\hat w})S(\bs{\hat w})=\IP{(\bs A-\sigma\bs C)^{-1}\bs d}{\bs 1}=b.
$$
Similar to the analysis of the open replicator equations in Section \ref{sec:3}, we conclude that if
\begin{equation}\label{eq3:5}
    b=\IP{(\bs A-\sigma\bs C)^{-1}\bs d}{\bs 1}\leq (\gamma e)^{-1}
\end{equation}
then it is possible to have one or two nonzero equilibria, which, under some additional conditions, may be inside the positive cone. We will use the same notation $\bs{\hat{w}}^1$ for the equilibrium that satisfies the condition $S(\bs{\hat w}^1)<\gamma^{-1}$ and $\bs{\hat w}^2$ for $S(\bs{\hat w}^2)>\gamma^{-1}$. If condition \eqref{eq3:5} does not hold than the origin is globally asymptotically stable.

The key conclusion from the above analysis is that, despite the inclusion of the deleterious virus macromolecules, it is still possible to have strictly positive equilibria $(\bs{\hat w},\bs{\hat v})$. It is difficult to say anything the stability of these equilibria in the full generality, but the numerical experiments show that existence of such equilibria, even if they are unstable, can be accompanied by the presence of positive invariant set $B$ inside the cone $\R^n_+\times \R^k_+$, in other words the system remains conditionally permanent, and hence the algorithm of evolutionary adaptation can be applied to this system (see below).

We still need an optimization criterion for system \eqref{eq3:1}. For this reason we introduce the following definition.

\begin{definition}Assume that system \eqref{eq3:1} is conditionally permanent with the positive equilibrium $(\bs{\hat w},\bs{\hat v})=(\bs{\hat w},\beta^{-1}\bs B^\top\bs{\hat w})$. Then its mean fitness at this equilibrium is defined as
\begin{equation}\label{eq3:6}
    \bar f(\bs{\hat w})=\frac{\IP{\bs{A\hat w}}{\bs{\hat w}}-\sigma \IP{\bs{C\hat w}}{\bs{\hat w}}}{\IP{\bs{D \hat w}}{\bs 1}}=\frac{\IP{(\bs A-\sigma \bs C)\bs{\hat w}}{\bs{\hat w}}}{\IP{\bs{D \hat w}}{\bs 1}}\,.
\end{equation}
The mean fitness of the subsystem in coordinates $\bs v$, which we call the virus, at the same positive equilibrium is defined as
\begin{equation}\label{eq3:7}
    \bar g(\bs{\hat w})=\IP{\bs{C\hat w}}{\bs{\hat w}}.
\end{equation}
\end{definition}
\section{Evolutionary adaptation as a result of interaction of replicator system with virus infection}\label{sec:4}
As defined above (see \eqref{eq3:1}), the interaction of replicator system with a virus infection results in the dynamics with opposite interests. We assume that the replicator system, following the basic tenet of evolutionary biology (but see also \cite{birch2016natural}) aims to maximize its average fitness \eqref{eq3:6}, whereas the virus population's goal is to maximize \eqref{eq3:7}. From a mathematical point of view such problem can be formulated in terms of \textit{differential game theory} (e.g., \cite{isaacs1999differential,petrosjan1993differential}). However, for such complicated system as \eqref{eq3:1}, even a numerical solution of such differential game is unfeasible, therefore we suggest a much simpler heuristic alternative approach. This approach is based on the following assumptions, which we assume to hold:
\begin{enumerate}
\item The system \eqref{eq3:1} is conditionally permanent.
\item The evolutionary changes are described by continuous changes of the matrices $\bs A$ and $\bs B$ within the given admissible sets.
\item The time scale in which the elements of matrices change is much slower then the time scale of the active dynamics of the system (for precise mathematical formulation of the separation of time scales see, e.g., \cite{bratus2024food}).
\end{enumerate}
In \cite{bratus2024food,bratus2018evolution,drozhzhin2021fitness} it is shown that if we assume that the assumptions above hold, then the process of evolutionary adaptation (in our particular case of competition between the replicator and virus) can be described with the help of the equations for the positive equilibrium point, the elements of which depend continuously on the slow evolutionary time $\tau = \varepsilon t$, where $t$ is the fast time of our system, and $\varepsilon>0$ is a small parameter.

Assume that matrices $\bs A(\tau), \bs C(\tau)$ as defined above are smooth functions of parameter $\tau$ laying within the admissible sets
\begin{equation}\label{eq4:1}
    \sum_{i,j=1}^na_{ij}^2\leq R_A,\quad \sum_{i,j=1}^n c_{ij}^2\leq R_C,\quad R_A,R_C>0.
\end{equation}
Note that alternatively we could constrain the elements of matrix $\bs B$ by some positive constant $R_B$, and this would imply the condition on the elements of $\bs C$, below in numerical computations we use this approach.

Let $\delta \bar f(\tau),\delta \bar g(\tau),\delta \bs A(\tau), \delta \bs C(\tau), \delta\bs{\hat w}(\tau),\delta\bs{\hat v}(\tau)$ denote the main linear parts of the increments of functions $\bar f(\tau), \bar g(\tau)$, of the elements of the matrices $\bs A(\tau), \bs C(\tau)$ and of the vectors $\bs{\hat{w}}(\tau),\bs{\hat v}(\tau)$ respectively if the slow time (the evolutionary parameter) changes from $\tau$ to $\Delta \tau$, where $\Delta \tau>0$ is small enough. The following proposition holds (the proof is given in Appendix).
\begin{proposition}\label{pr4:1}Assume that at some evolutionary time $\tau$ system \eqref{eq3:1} is conditionally permanent with the positive equilibrium point $(\bs{\hat w}(\tau),\bs{\hat v}(\tau))$. Let $\bs R(\tau)=\bs A(\tau)-\sigma \bs C(\tau)$.

Then if $\delta \bs B(\tau)=0$ then the main linear part $\delta \bar f(\tau)$ of the increment of $\bar f(\tau)$ is a linear functional of the elements of $\delta\bs A(\tau)$, and is given by
\begin{equation}\label{eq4:2}
    \delta \bar f_A(\tau)=\frac{\gamma}{\psi(\bs{\hat w}(\tau))}\frac{\IP{\bs R^{-1}(\tau)\delta \bs A(\tau)\bs{\hat{w}}(\tau)}{\bs 1}}{\gamma S(\bs{\hat w}(\tau))-1}\,.
\end{equation}

If $\delta\bs A(\tau)=0$ then the main linear part $\delta \bar f(\tau)$ of the increment of $\bar f(\tau)$ is a linear functional of the elements of $\delta\bs C(\tau)$, and is given by
\begin{equation}\label{eq4:3}
    \delta \bar f_B(\tau)=\frac{\sigma}{\psi(\bs{\hat w}(\tau))}\frac{\IP{\bs R^{-1}(\tau)\delta \bs C(\tau)\bs{\hat{w}}(\tau)}{\bs 1}}{1-\gamma S(\bs{\hat w}(\tau))}\,.
\end{equation}
\end{proposition}

The constrains \eqref{eq4:1} and the expressions for the linear part of the average fitness allow us to recast the process of evolutionary competition between the replicator equation and virus into a sequence of linear programming problems, which we interpret in terms of ``attack'' and ``defence'' steps.

Let, first, assume that the virus does not change ($\delta \bs B(\tau)=0$) whereas the replicator evolves from the state $\bs A(\tau)$ to the state $\bs A(\tau)+\delta\bs A(\tau)$. Then from the first constraint in \eqref{eq4:1} it follows that
\begin{equation}\label{eq4:4}
    \sum_{i,j=1}^n a_{ij}(\tau)\delta a_{ij}(\tau)\leq 0,\quad |\delta a_{ij}(\tau)|\leq \delta_A.
\end{equation}
Since it is important to keep the positive equilibrium, we also require that
\begin{equation}\label{eq4:5}
    \bs{\hat w}(\tau)+\delta \bs{\hat{w}}(\tau)\geq 0.
\end{equation}
Therefore, on each small step $\Delta\tau$ we get the linear programming problem of maximizing $\delta\bar f_A(\tau)$ in \eqref{eq4:2} with all possible $\delta \bs A(\tau)$ are subject to the linear constraints \eqref{eq4:4} and \eqref{eq4:5}. Within the general scheme of consumer--resource interaction we call it the ``defence'' step.

Now we assume that the elements of $\bs A(\tau)$ do not change, $\delta\bs A(\tau)=0$, whereas $\delta\bs B(\tau)\neq 0$. This is the ``attack'' step. Similarly to the previous we consider the constraints
\begin{equation}\label{eq4:6}
    \sum_{i=1}^n\sum_{j=1}^k b_{ij}\delta b_{ij}(\tau)\leq 0,\quad |\delta b_{ij}(\tau)|\leq \delta_B.
\end{equation}
Now the linear programming problem reads: minimize the linear functional $\delta \bar f_B(\tau)$ in \eqref{eq4:3} subject to the linear constraints \eqref{eq4:6} and \eqref{eq4:5}.

Alternating the steps of ``defence'' and ``attack'' at each step $\Delta \tau$ we obtain a sequence of changes in the matrices $\bs A$ and $\bs B$, which correspond to the defence--attack process. It is important to mention that at each step of this process we can return to the dynamical equations \eqref{eq3:1} with the new elements of matrices $\bs A$ and $\bs B$ and study, at least numerically, the time dependent solutions, which may be quite different from the coordinates of the positive equilibrium. As the following examples show, our consumer--resource evolutionary mathematical model, in addition to the detailed internal structure, is capable of producing the generic evolutionary outcomes mentioned in the introduction.
\section{Examples}\label{sec:5}
In the first two examples we take $n=5, k=3$,
$$
\bs B=\begin{bmatrix}
        0 & 0 & 0 \\
        0 & 0 & 0 \\
        0.0896 & 0.0864 & 0 \\
        0 & 0.0278 & 0.0761 \\
        0.0944 & 0 & 0.0906 \\
      \end{bmatrix},\quad \bs d=(0.0648, 0.0577, 0.0444, 0.0439, 0.0468),
$$
$\gamma = 0.1,\,\alpha=1,\beta =0.1$.

The key difference in the examples below is that, by changing the constants $\delta_A, \delta_B$, which place the constrains on the changes of the sizes of elements of matrices $\bs A$ and $\bs B$ respectively, we can give an evolutionary advantage to either replicator or virus.

\begin{example}[Evolutionary coexistence in the consumer--resource system]\label{ex4:1}Here we take $\delta_A=10^{-5},\,\delta_B=10^{-4}$, which gives a strong evolutionary advantage to the virus. Despite this fact, as numerical experiments show, the dynamics of the replicator and virus still preserve the important properties present at the initial evolutionary time $\tau=0$. The system remains to be conditionally permanent and demonstrates long time coexistence (see Fig. \ref{fig4:1}, especially panels $(e)$ and $(f)$), with periodic oscillations of both the replicator and virus concentrations.
\begin{figure}[!th]
\centering
\includegraphics[width=0.495\textwidth]{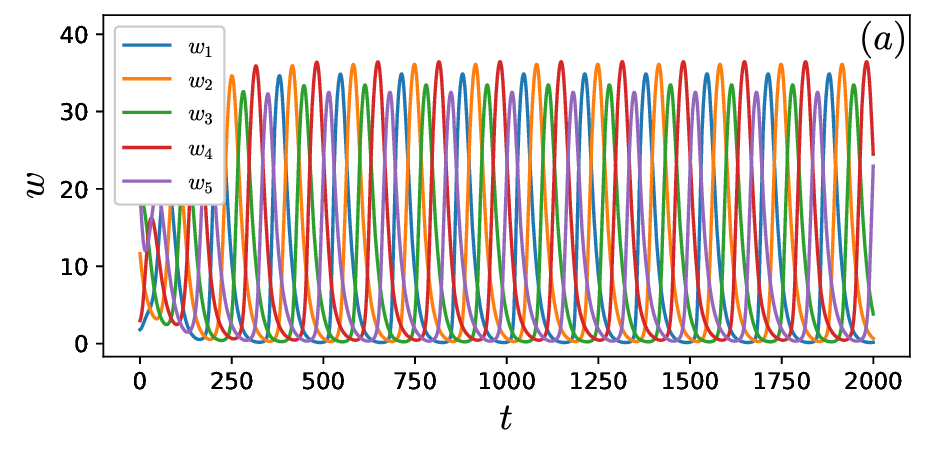}
\includegraphics[width=0.495\textwidth]{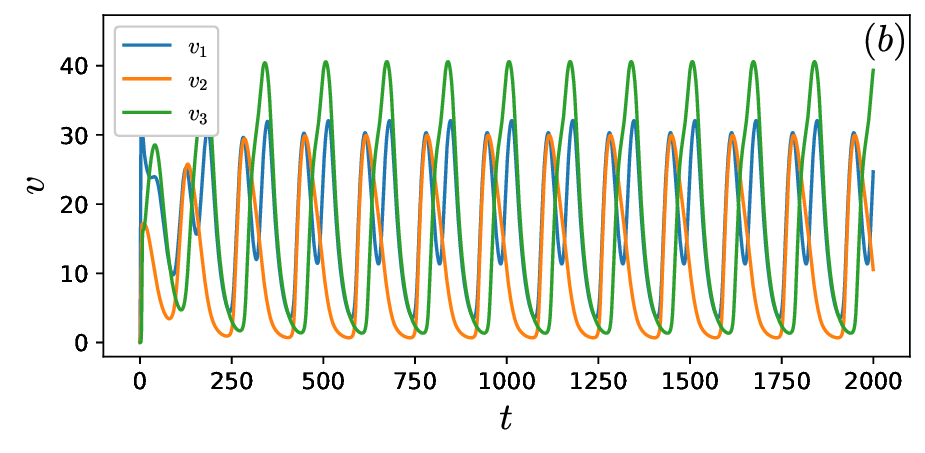}
\includegraphics[width=0.495\textwidth]{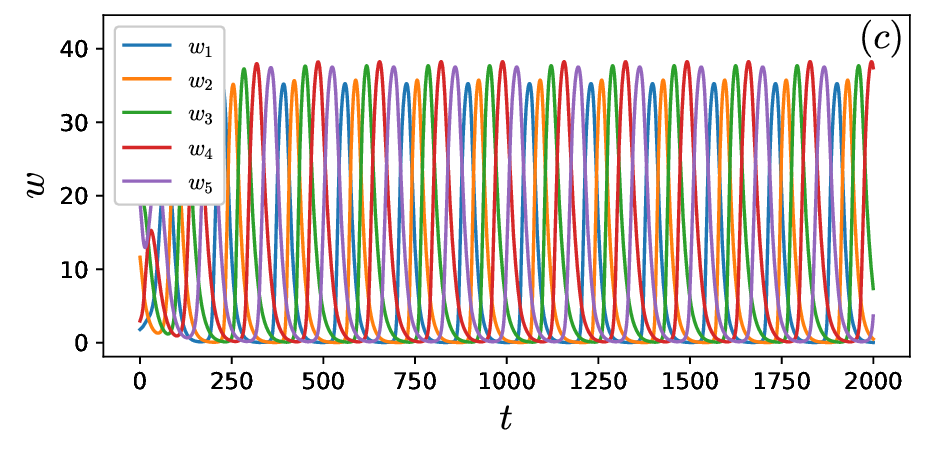}
\includegraphics[width=0.495\textwidth]{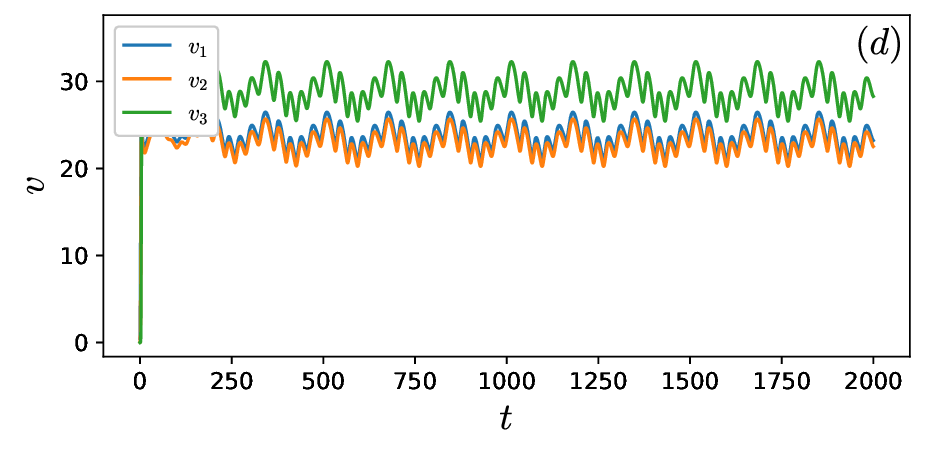}
\includegraphics[width=0.495\textwidth]{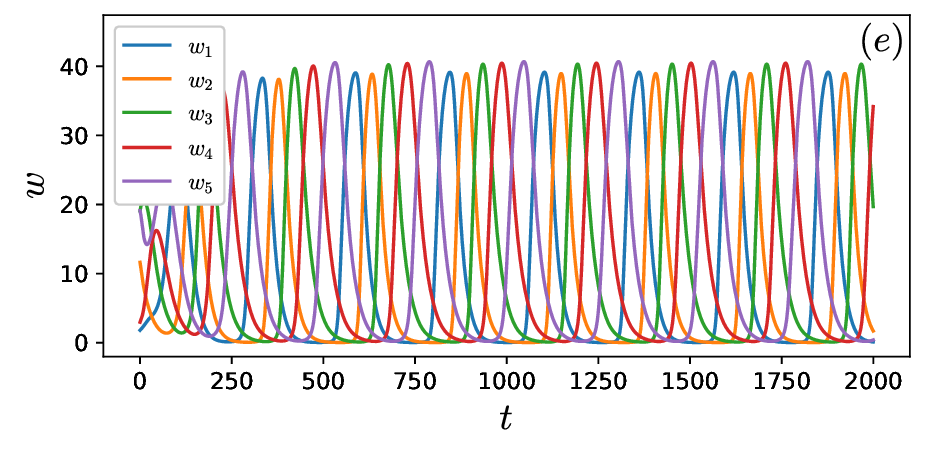}
\includegraphics[width=0.495\textwidth]{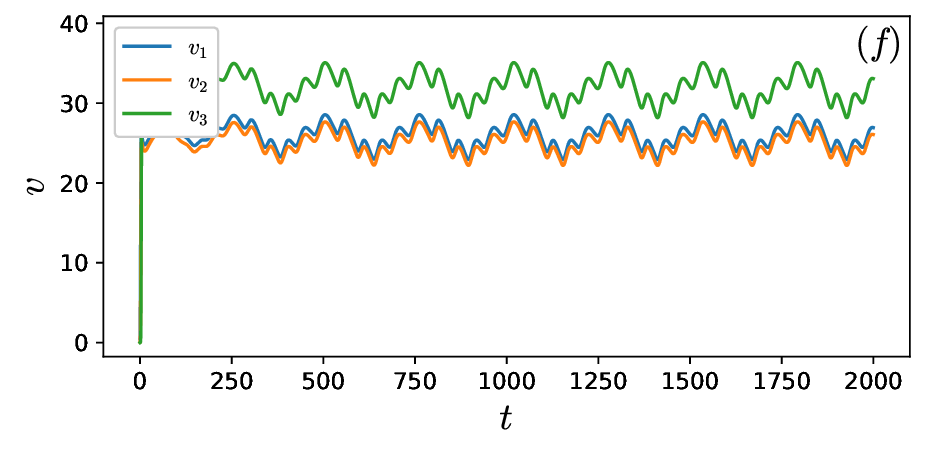}
\caption{Coexistence of replicator--virus system in Example \ref{ex4:1}. In this case the virus has an evolutionary advantage. $(a),(b)$ --- the dynamics of replicator and virus at evolutionary time $\tau=0$. $(c), (d)$ --- the dynamics of replicator and virus after 5000 evolutionary steps. $(e), (f)$ --- the dynamics of replicator and virus after 20000 evolutionary steps.}\label{fig4:1}
\end{figure}

The mean fitness of both the replicator and the virus in this experiment continue to increase with evolutionary time (as can be seen in Fig. \ref{fig4:2}, although initially the mean fitness of the replicator fast decreases, Fig. \ref{fig4:2}$(a)$). Moreover, none of the coordinates of the positive equilibrium approaches zero, Fig. \ref{fig4:3}. This is especially interesting since the general trend on longer evolutionary distances for replicator equations is to evolve most of the equilibrium coordinates close to zero, such that only one or two of them are far from zero (the obvious case of the evolution of the fittest, see many examples in \cite{bratus2018evolution,drozhzhin2021fitness}).
\begin{figure}[!th]
\centering
\includegraphics[width=0.495\textwidth]{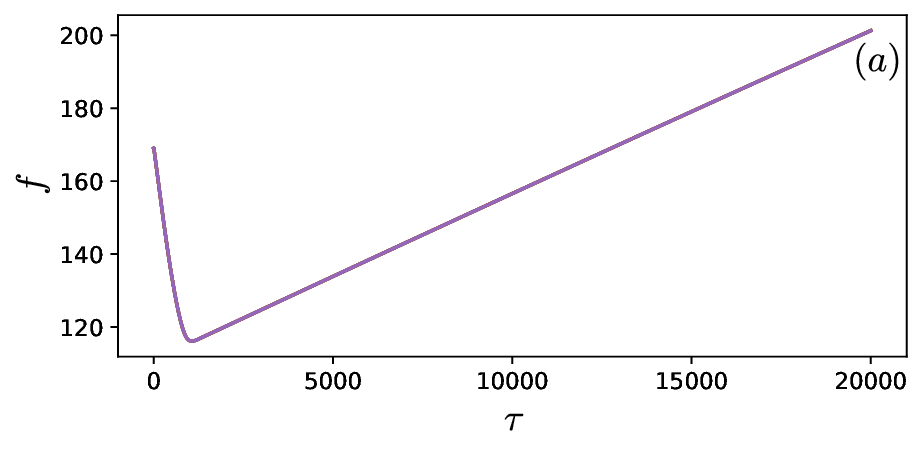}
\includegraphics[width=0.495\textwidth]{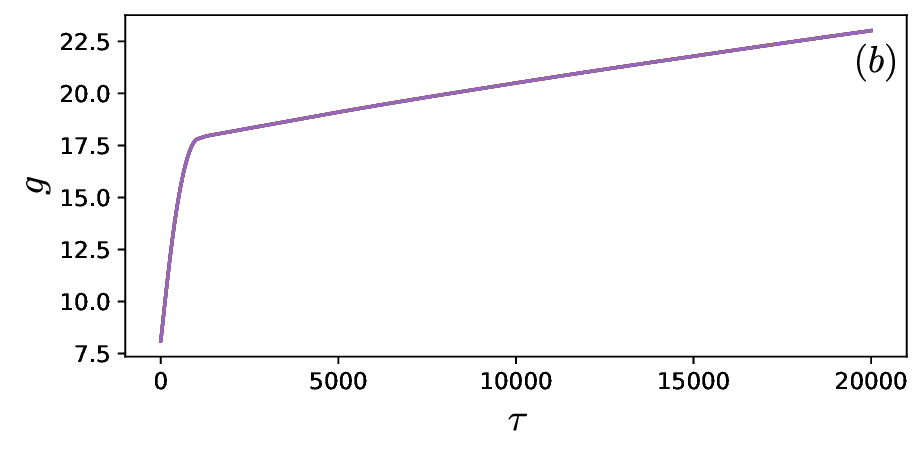}
\caption{Changes of the average fitness of replicator $(a)$ and virus $(b)$ versus evolutionary steps  in Example \ref{ex4:1}.}\label{fig4:2}
\end{figure}

\begin{figure}[!th]
\centering
\includegraphics[width=0.495\textwidth]{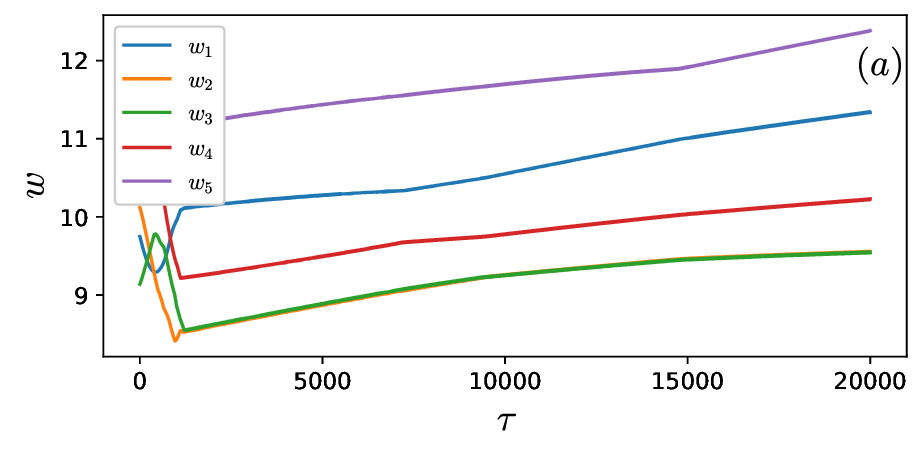}
\includegraphics[width=0.495\textwidth]{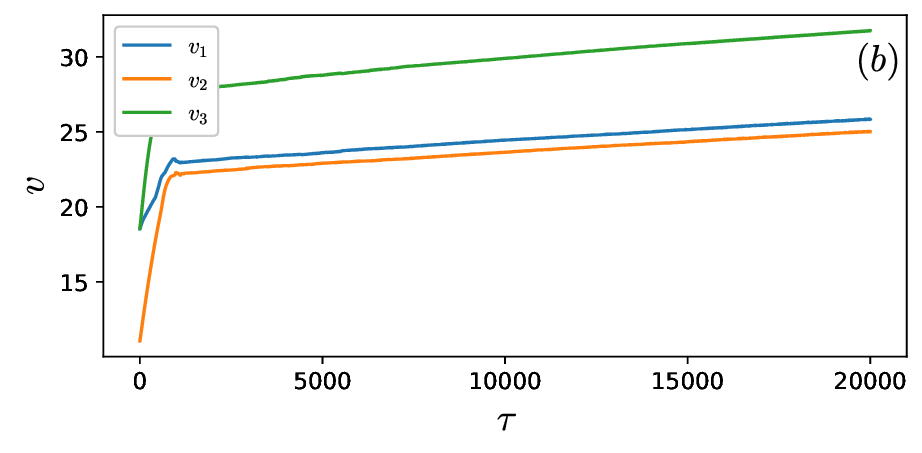}
\caption{Changes of the positive equilibria coordinates of replicator $(a)$ and virus $(b)$ versus evolutionary steps in Example \ref{ex4:1}.}\label{fig4:3}
\end{figure}

In this particular example, despite the opposite evolutionary goals of the replicator (which strives to reduce the damage caused by the parasite) and the virus (which evolves its ability to better attack the replicator on average), the system finds itself in a long term quasi-stationary state, which is acceptable to both conflict parties.

It is interesting to note that in all the examples when we observe long time coexistence, the interaction network of the replicator (the cell) becomes more and more interconnected. For instance, in Fig. \ref{fig5:1} we show the original hypercycle network at the beginning of the simulation experiment and the final state; it can be seen that after multiple steps of fitness optimization of both the replicator and the virus, the number of connections has significantly increased. We remark that similar effects in our prior work \cite{bratus2018evolution} led to the important phenomenon of the resistance to parasite invasion in replicator networks.
\begin{figure}
\centering
\includegraphics[width=0.4\textwidth]{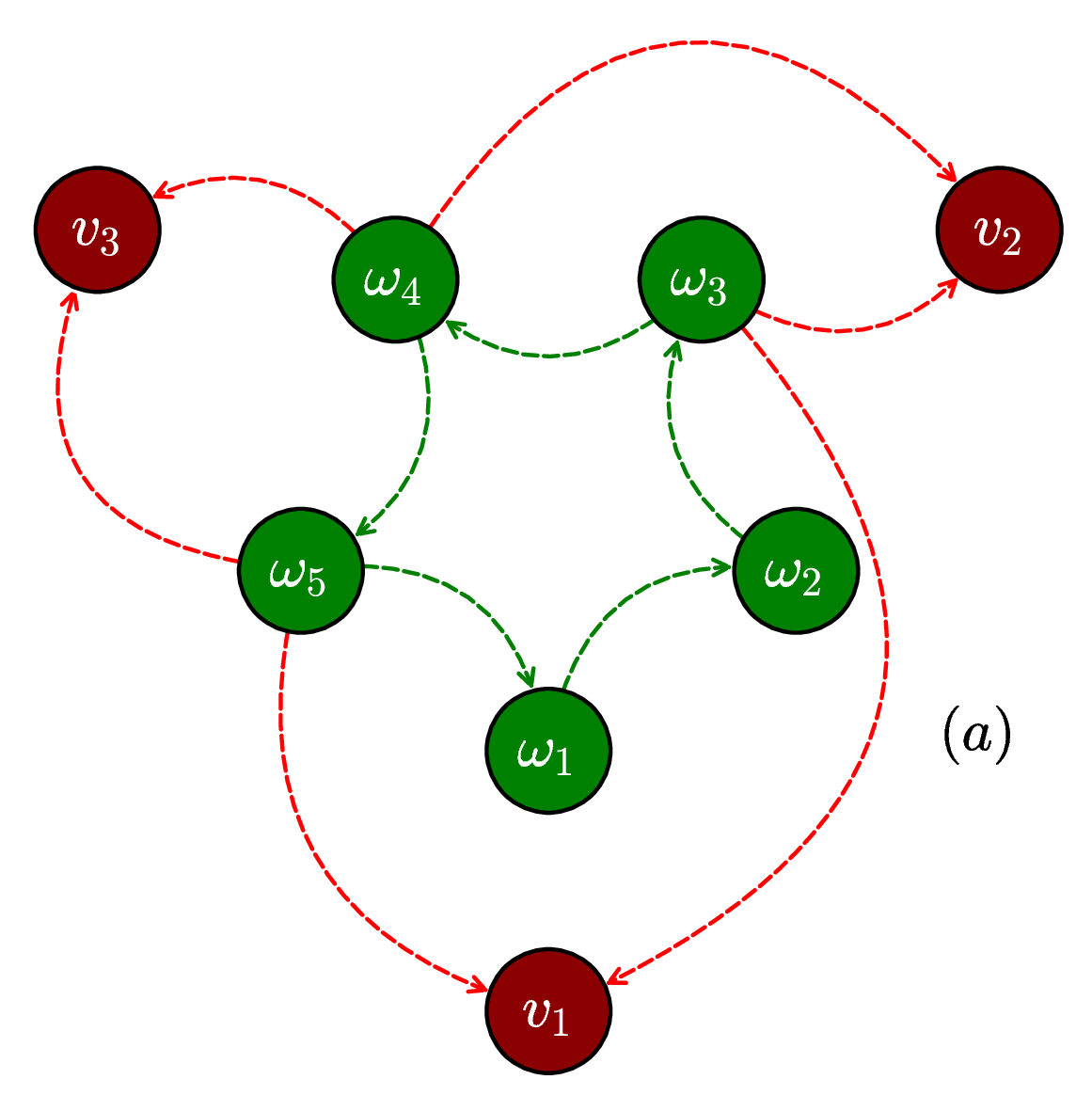}
\includegraphics[width=0.55\textwidth]{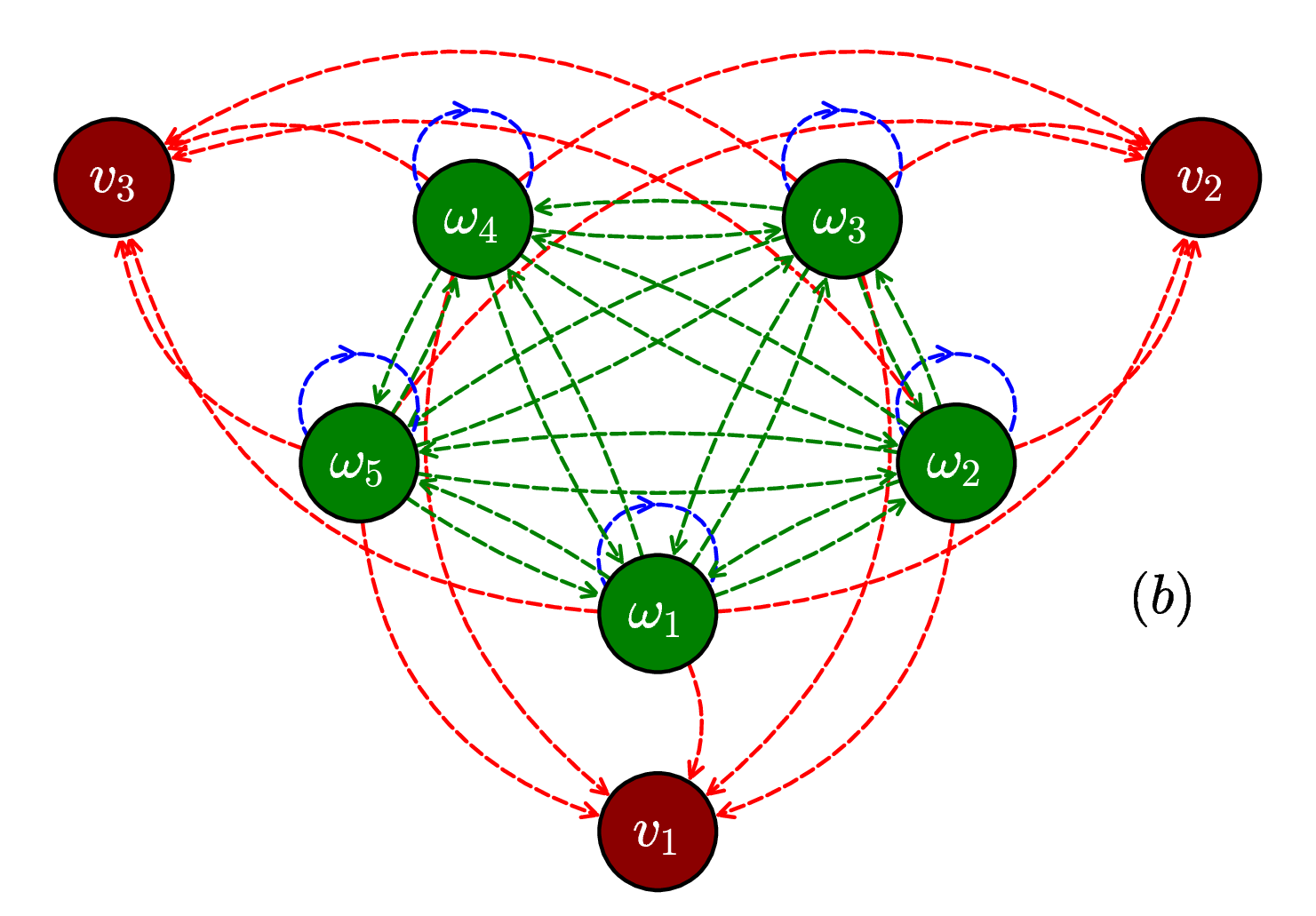}
\caption{The interaction network at the initial time point (left) and at the end point (right) in Example \ref{ex4:1}.}\label{fig5:1}
\end{figure}
\end{example}
\newpage

\begin{example}[Replicator prospers with the presence of the virus]\label{ex4:2}
In this example we take $\delta_A=10^{-3},\delta_B=10^{-4}$, which gives an evolutionary advantage to the replicator. The evolution of the replicator follows more or less the scenario, which would be expected in the situation if there was no virus present: after a sufficiently long evolutionary time only one of the replicators (the fittest one in this situation) is present (Fig. \ref{fig4:4}$(e)$). At same time, despite apparent adaptation of the replicator, the virus also survives (Fig. \ref{fig4:4}$(f)$). In other words, the complex system of both inter- and intra- population interactions leads to the situation which looks for an outside observer as a mutualistic system: both parties, quite surprisingly, benefit in the long run despite the opposite evolutionary goals.

\begin{figure}[!th]
\centering
\includegraphics[width=0.495\textwidth]{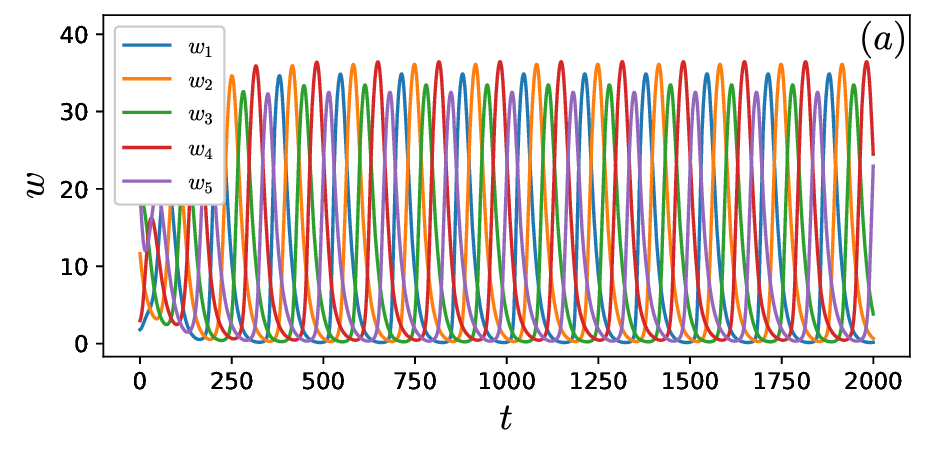}
\includegraphics[width=0.495\textwidth]{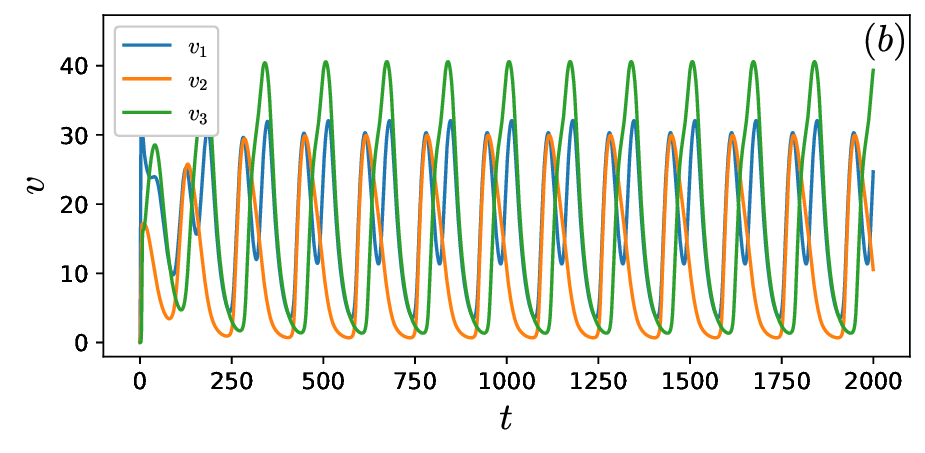}
\includegraphics[width=0.495\textwidth]{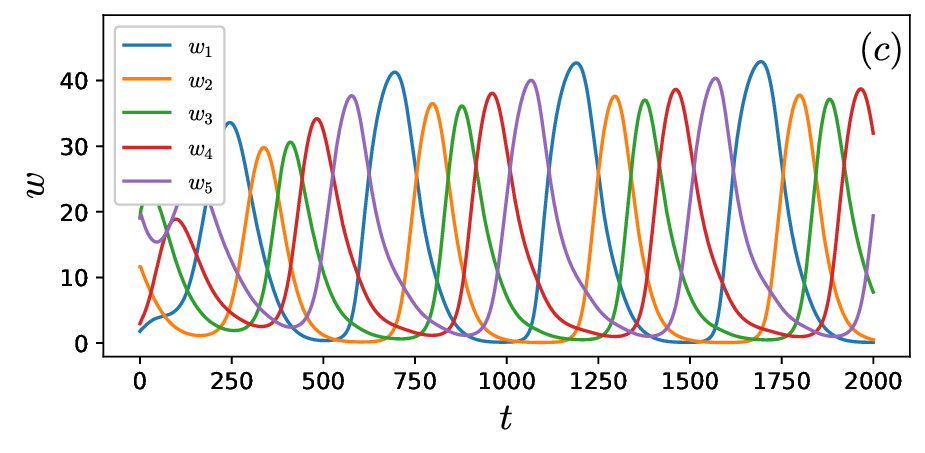}
\includegraphics[width=0.495\textwidth]{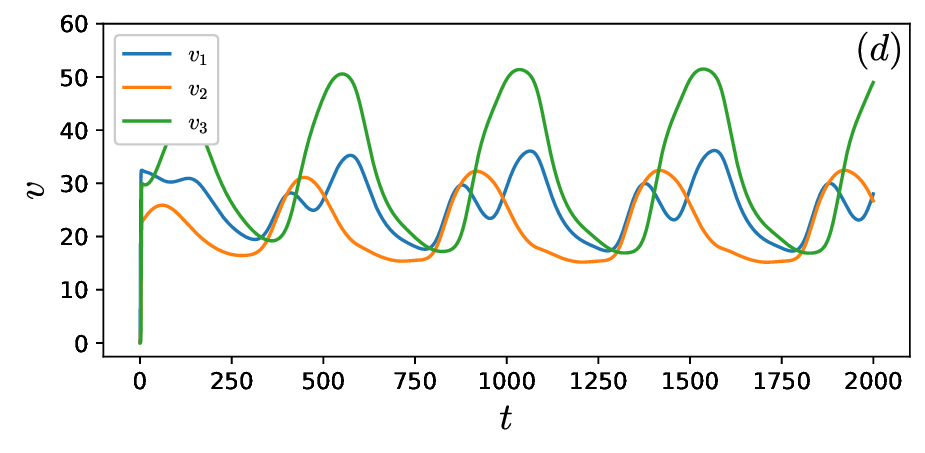}
\includegraphics[width=0.495\textwidth]{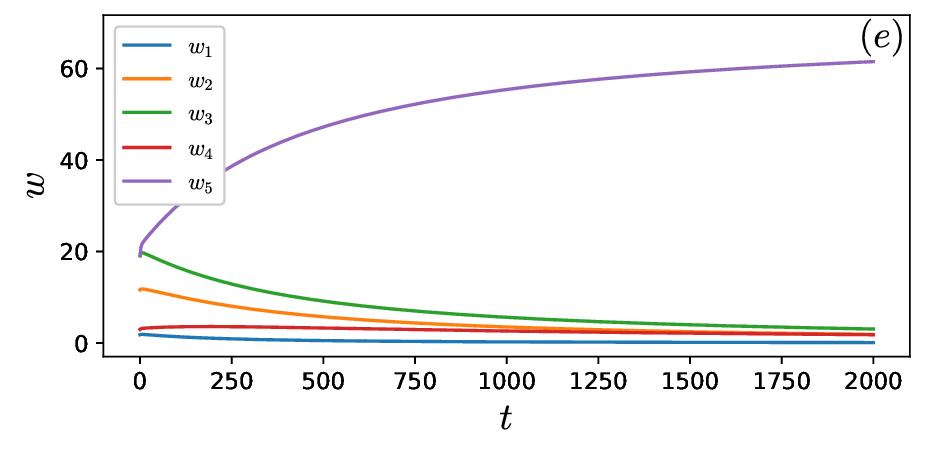}
\includegraphics[width=0.495\textwidth]{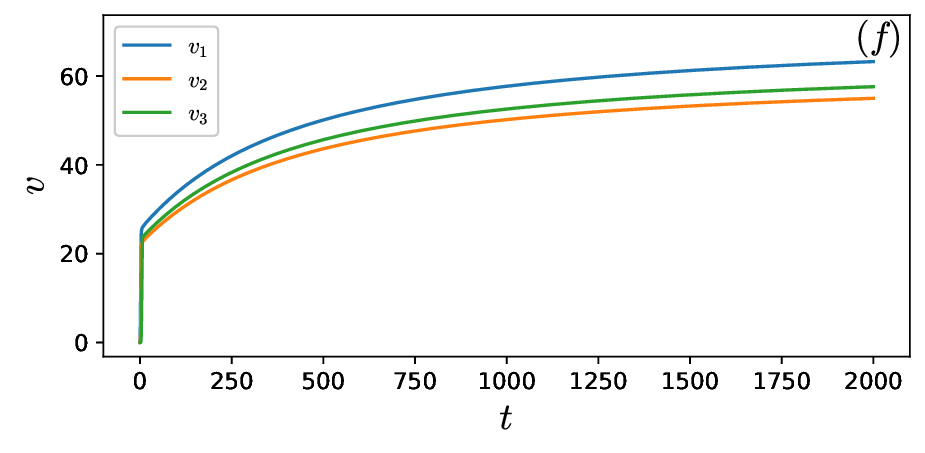}
\caption{Coevolutionary dynamics of replicator--virus system in Example \ref{ex4:2}. In this case the replicator is given an evolutionary advantage. $(a),(b)$ --- the dynamics of replicator and virus at evolutionary time $\tau=0$. $(c), (d)$ --- the dynamics of replicator and virus after 500 evolutionary steps. $(e), (f)$ --- the dynamics of replicator and virus after 2000 evolutionary steps.}\label{fig4:4}
\end{figure}

The corresponding changes in the average fitness and coordinates of positive equilibria are shown in Figs. \ref{fig4:5} and \ref{fig4:6} respectively.

\begin{figure}[!th]
\centering
\includegraphics[width=0.495\textwidth]{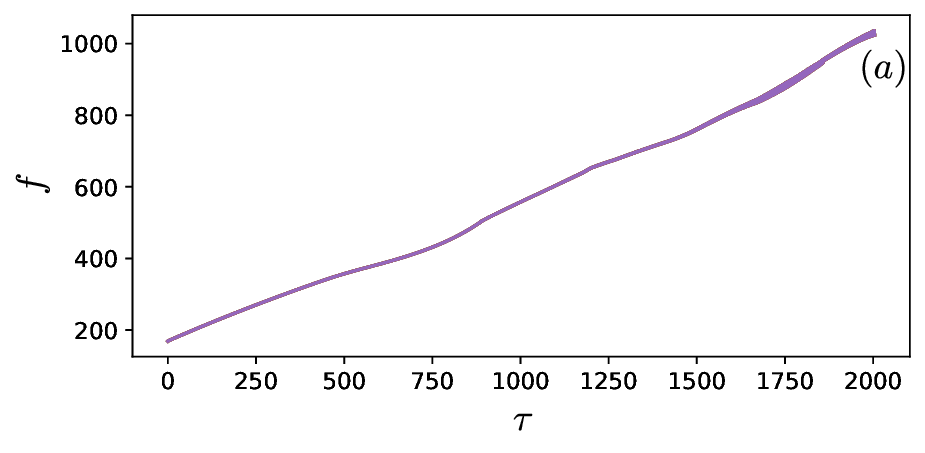}
\includegraphics[width=0.495\textwidth]{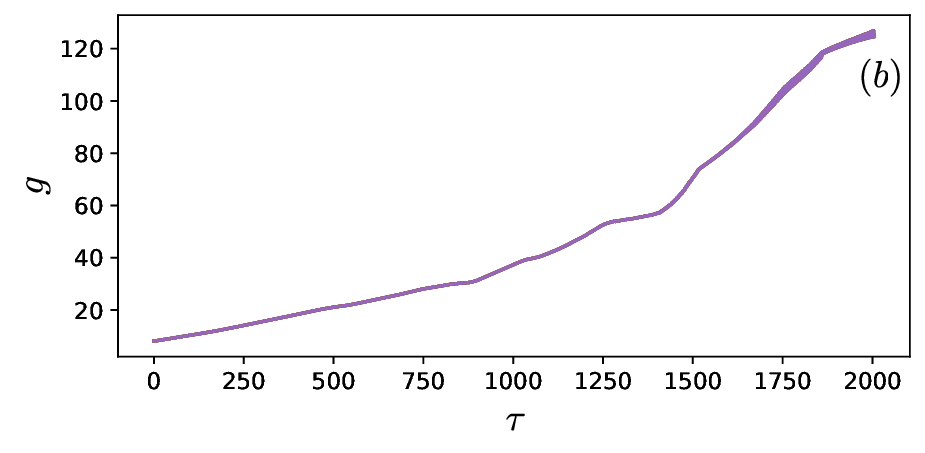}
\caption{Changes of the average fitness of replicator $(a)$ and virus $(b)$ versus evolutionary steps  in Example \ref{ex4:2}.}\label{fig4:5}
\end{figure}

\begin{figure}[!th]
\centering
\includegraphics[width=0.495\textwidth]{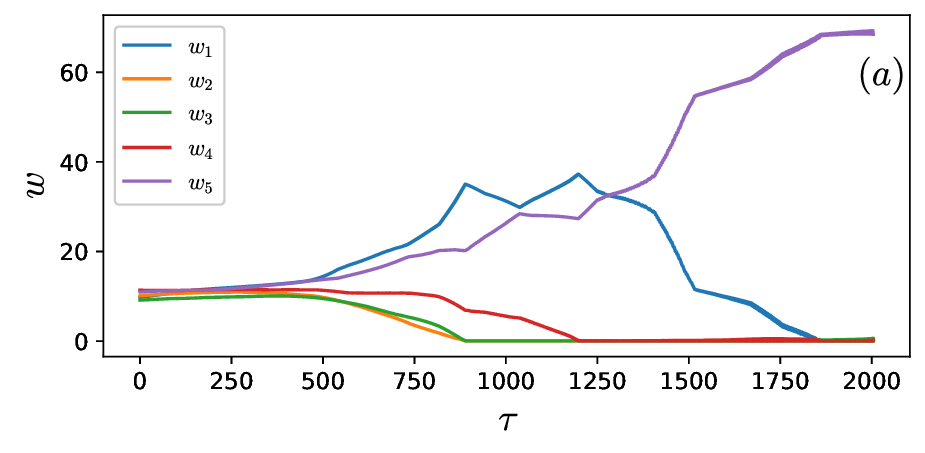}
\includegraphics[width=0.495\textwidth]{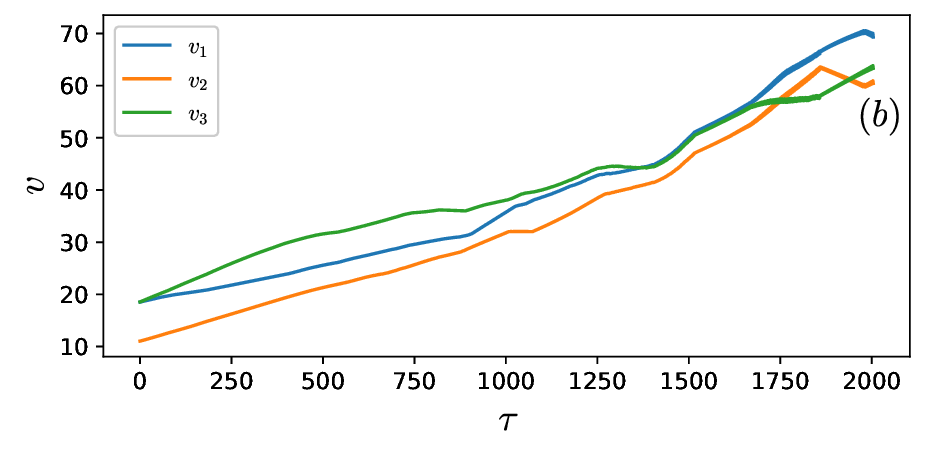}
\caption{Changes of the positive equilibria coordinates of replicator $(a)$ and virus $(b)$ versus evolutionary steps in Example \ref{ex4:2}.}\label{fig4:6}
\end{figure}

Note that in the considered example the behavior of average fitness is monotone for both the replicator and the virus.
\end{example}

\pagebreak

In the next two examples we take $k=1$, i.e., we consider only one virus species in the system to highlight the results (similar phenomena occur in high dimensional systems, not shown).

\begin{example}[The replicator dominates the virus]\label{ex4:3}Consider an example when eventually, in the result of evolutionary adaptation, the replicator system evolves to the state, in which the virus is doomed to go extinct. To make it explicit, for this example we change the initial matrix $\bs B$ to
$$
\bs B=[0,0,0.1,0,0],
$$
i.e., we choose $k=1$, and set $\delta_A=10^{-3},\delta_B=10^{-6}$ (significant evolutionary advantage of the replicator). The results of the evolutionary algorithm are shown in Fig. \ref{fig4:7}. The evolution of the replicator is similar to the one observed in the previous example: all but one replicators disappear, leaving the evolutionary space for the fittest one (Fig. \ref{fig4:7}$(e)$). The virus, however, cannot adapt as fast in the given situation, and
eventually evolves to the state, in which highly adapted replicator cannot support its existence. As a result, the virus population tends to zero (Fig. \ref{fig4:7}$(f)$). The average fitness of the virus, increasing for a short initial time period, very fast becomes decreasing, with the second derivative being negative (see Fig. \ref{fig4:8}$(b)$).
\begin{figure}[!th]
\centering
\includegraphics[width=0.495\textwidth]{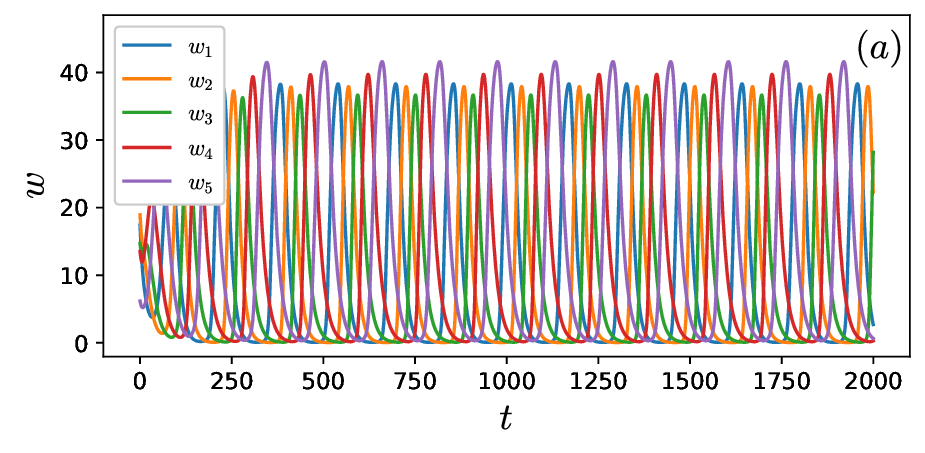}
\includegraphics[width=0.495\textwidth]{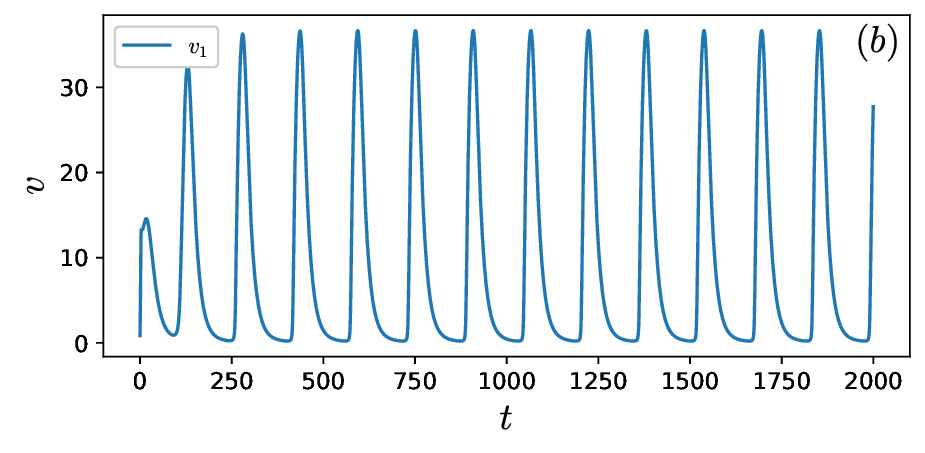}
\includegraphics[width=0.495\textwidth]{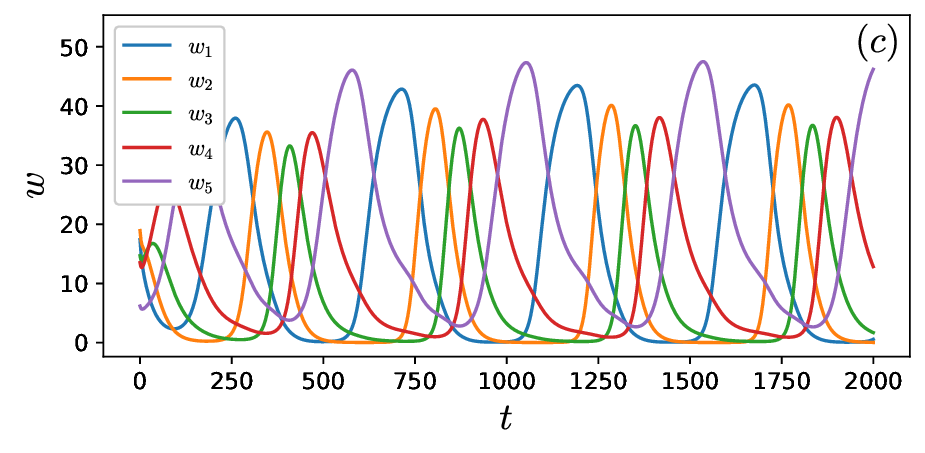}
\includegraphics[width=0.495\textwidth]{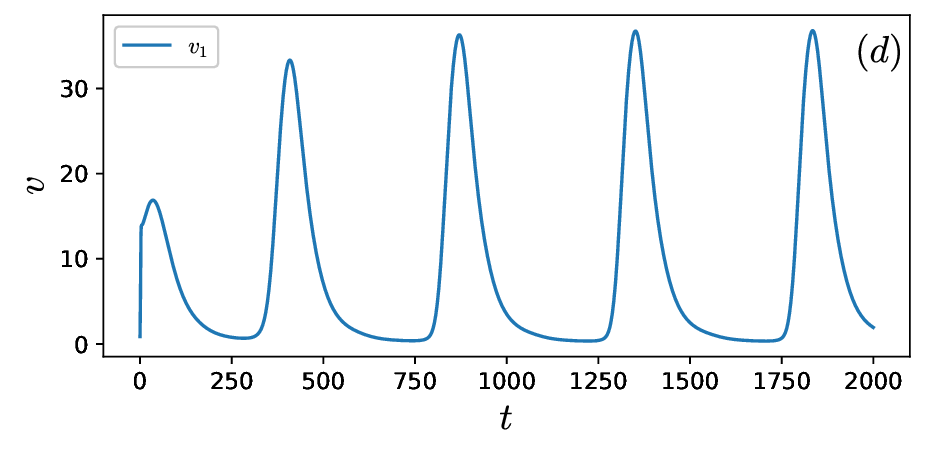}
\includegraphics[width=0.495\textwidth]{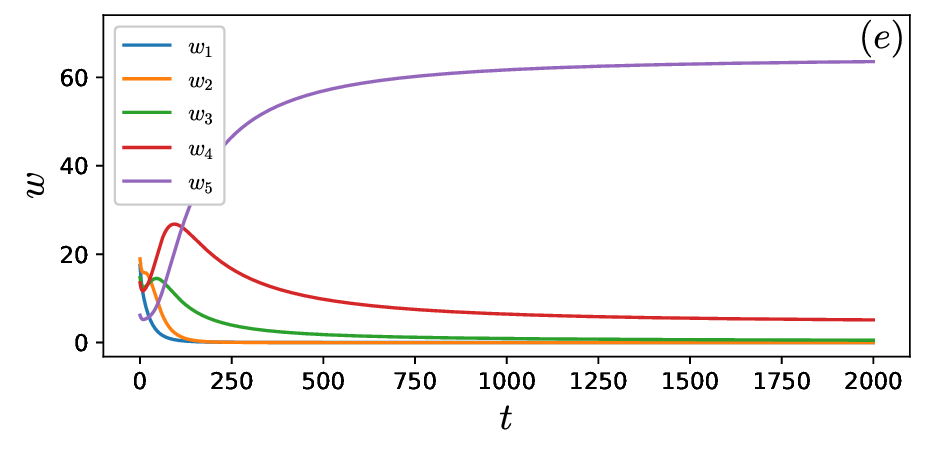}
\includegraphics[width=0.495\textwidth]{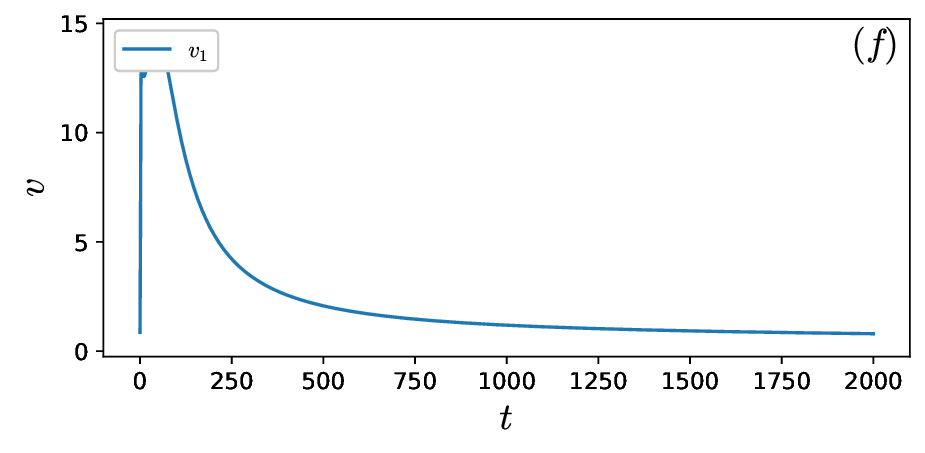}
\caption{Coevolutionary dynamics of replicator--virus system in Example \ref{ex4:3}. In this case the replicator is given a strong evolutionary advantage. $(a),(b)$ --- the dynamics of replicator and virus at evolutionary time $\tau=0$. $(c), (d)$ --- the dynamics of replicator and virus after 500 evolutionary steps. $(e), (f)$ --- the dynamics of replicator and virus after 1500 evolutionary steps.}\label{fig4:7}
\end{figure}

The corresponding changes in the average fitness and coordinates of positive equilibria are shown in Fig. \ref{fig4:8} and \ref{fig4:9} respectively.

\begin{figure}[!th]
\centering
\includegraphics[width=0.495\textwidth]{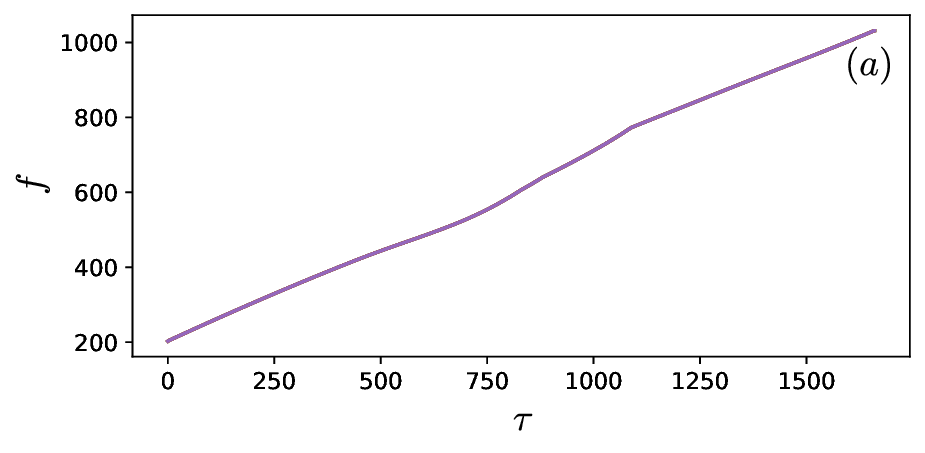}
\includegraphics[width=0.495\textwidth]{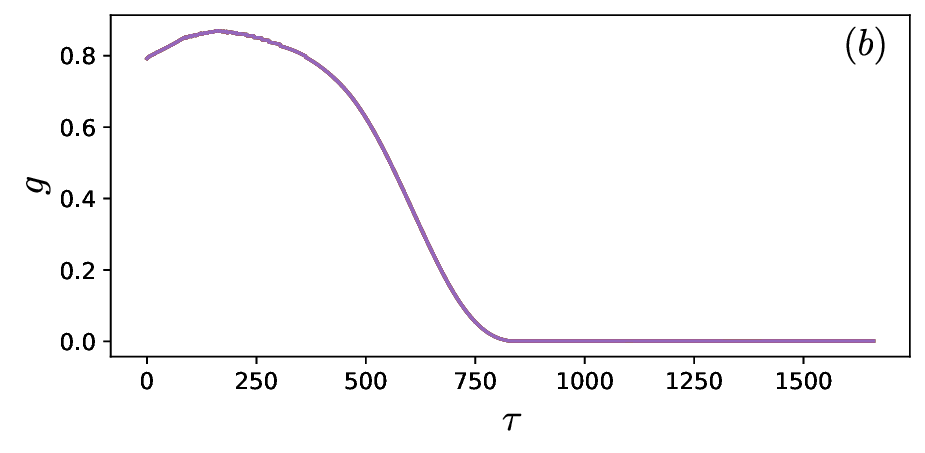}
\caption{Changes in the average fitness of replicator $(a)$ and virus $(b)$ versus evolutionary steps  in Example \ref{ex4:3}.}\label{fig4:8}
\end{figure}

\begin{figure}[!th]
\centering
\includegraphics[width=0.495\textwidth]{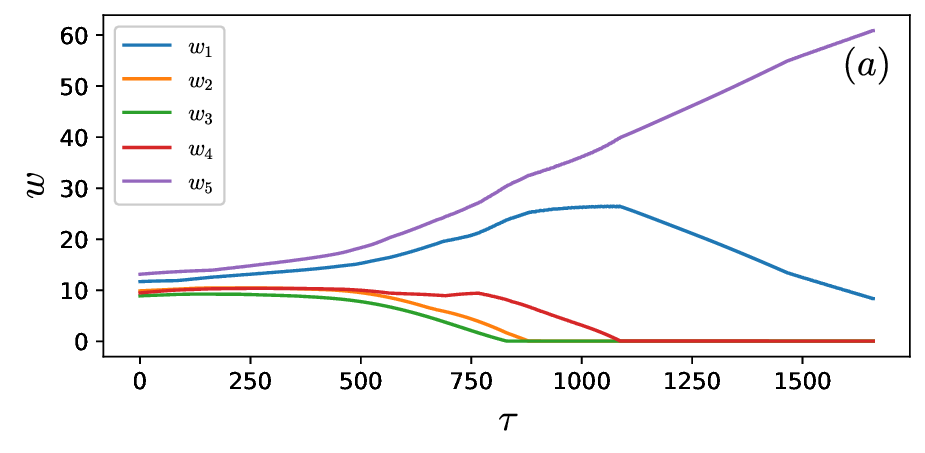}
\includegraphics[width=0.495\textwidth]{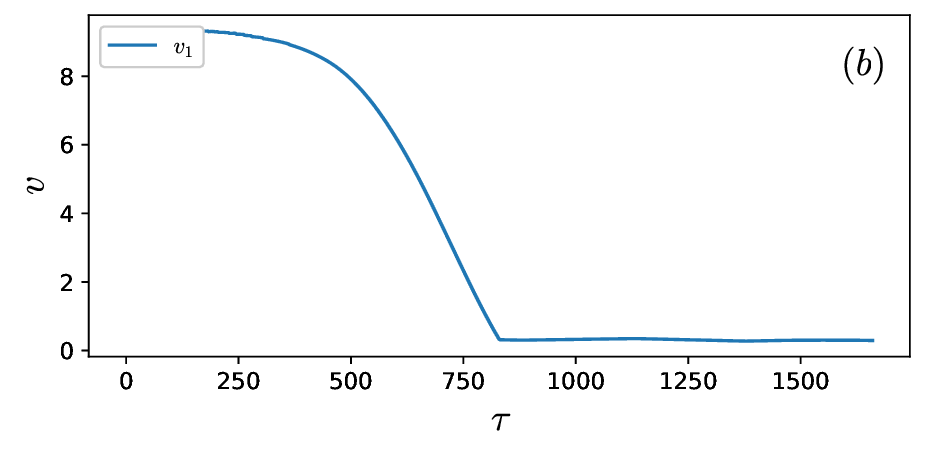}
\caption{Changes in the positive equilibria coordinates of replicator $(a)$ and virus $(b)$ versus evolutionary steps in Example \ref{ex4:3}.}\label{fig4:9}
\end{figure}

\end{example}

\pagebreak

\begin{example}[Virus dominates the replicator]\label{ex4:4}Finally consider the situation, in which we provide a strong evolutionary advantage to the virus. Here, as in the previous example, $k=1$, and $\delta_A=10^{-4},\delta_B=10^{-2}$. In other words, we alow the virus to make much more significant evolutionary steps compared with the replicator steps. In the long run such situation actually leads to eventual virus extinction, since very quickly it bring the average population fitness of the replicator to the point, at which it cannot survive (see Fig. \ref{fig4:10}, in which the changes in average fitness are shown, and Fig. \ref{fig4:11}, in which the evolution of the coordinates of the positive equilibrium is shown). In mathematical terms this particular scenario corresponds to the situation when the system stops being conditionally permanent; the ultimate outcome is total population extinction: $\bs{w}(t),\bs v(t)\to 0$ as $t\to\infty$.

\begin{figure}[!th]
\centering
\includegraphics[width=0.495\textwidth]{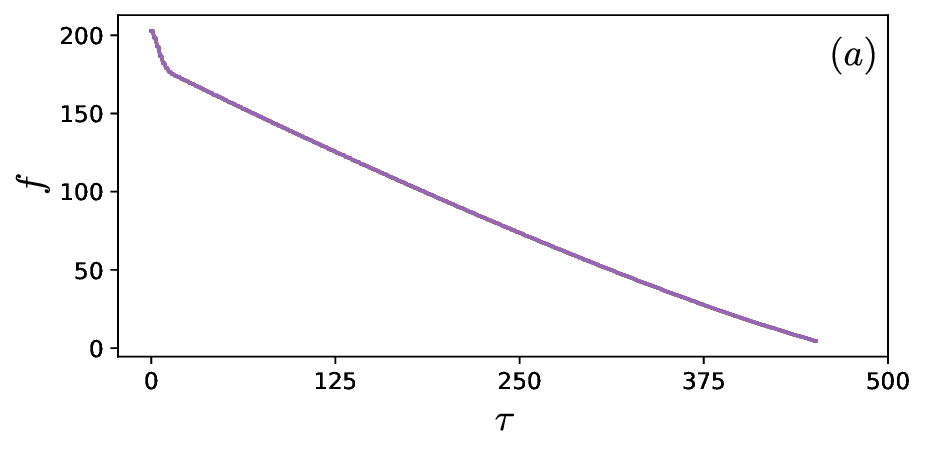}
\includegraphics[width=0.495\textwidth]{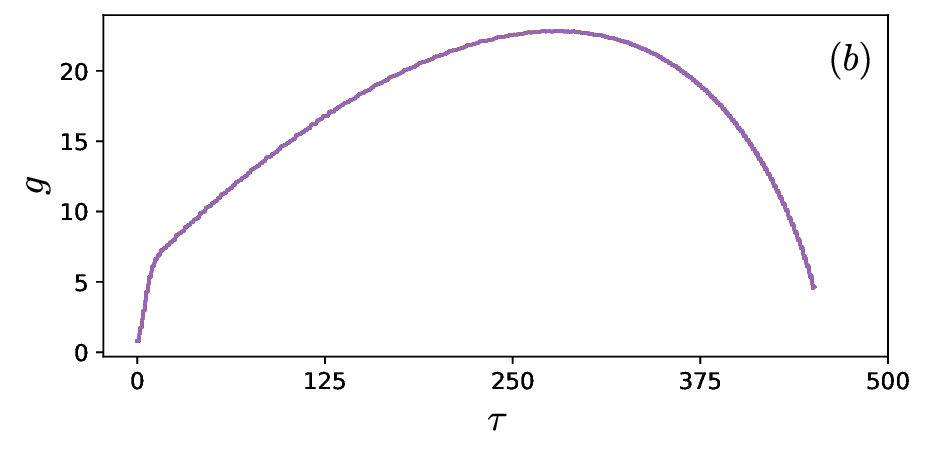}
\caption{Changes in the average fitness of replicator $(a)$ and virus $(b)$ versus evolutionary steps  in Example \ref{ex4:4}.}\label{fig4:10}
\end{figure}

\begin{figure}[!th]
\centering
\includegraphics[width=0.495\textwidth]{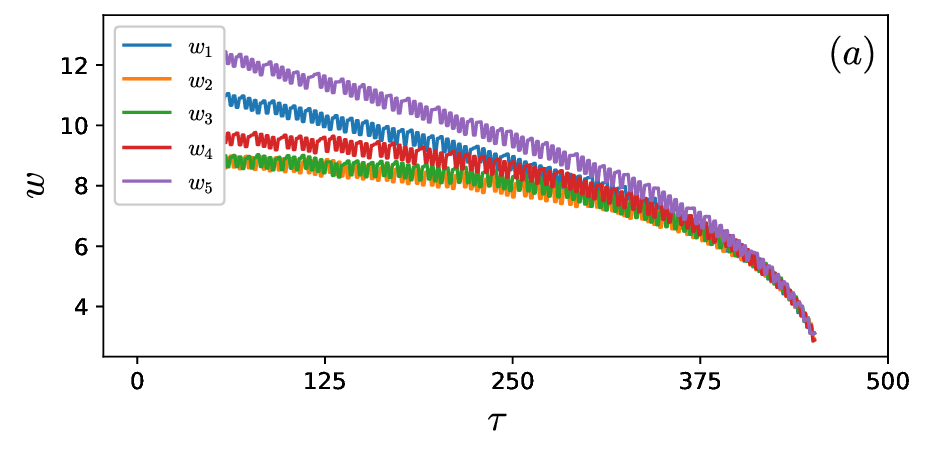}
\includegraphics[width=0.495\textwidth]{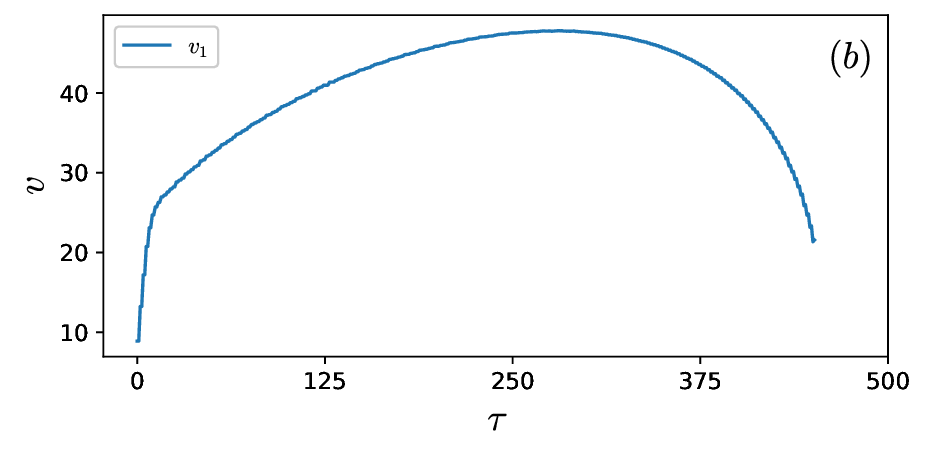}
\caption{Changes in the positive equilibria coordinates of replicator $(a)$ and virus $(b)$ versus evolutionary steps in Example \ref{ex4:4}.}\label{fig4:11}
\end{figure}

The time dependent dynamics is shown in Fig. \ref{fig4:12}. It is interesting to note that even at the previous to the last evolutionary step the system still exhibits a coexistence pattern. For an outside observer the situation does not look that dramatic (Fig. \ref{fig4:12}$(e),(f)$)). And yet one more evolutionary step (see Fig. \ref{fig4:12}$(g),(h)$) destroys the positive equilibrium completely, both the replicator and virus are attracted to zero.
\begin{figure}[!th]
\centering
\includegraphics[width=0.495\textwidth]{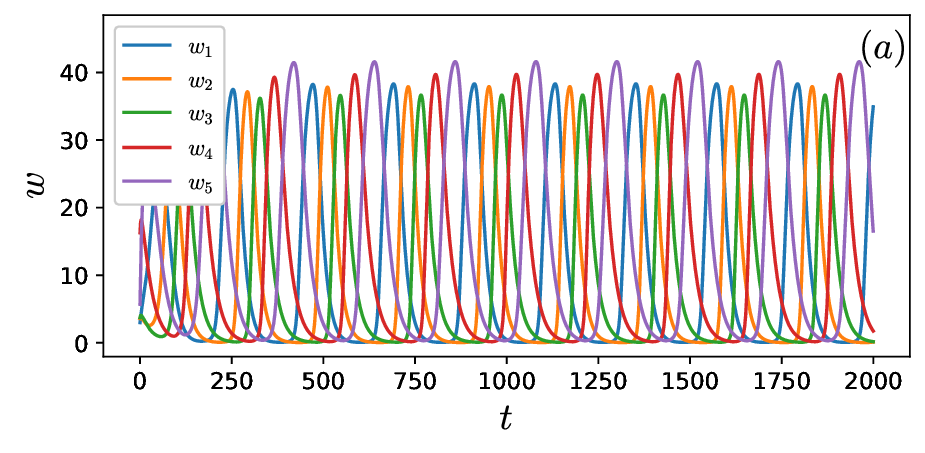}
\includegraphics[width=0.495\textwidth]{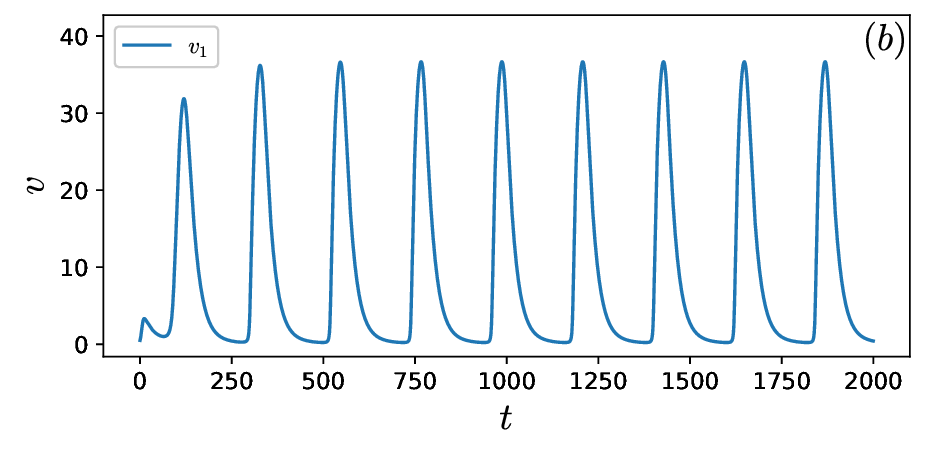}
\includegraphics[width=0.495\textwidth]{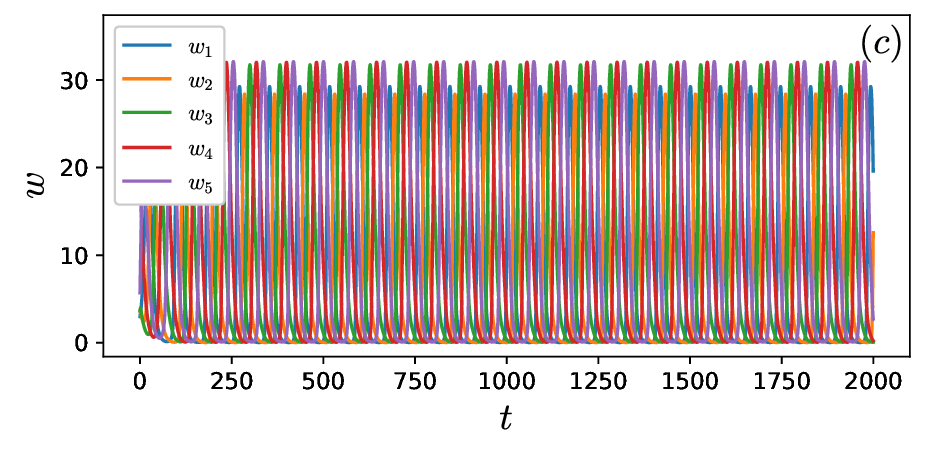}
\includegraphics[width=0.495\textwidth]{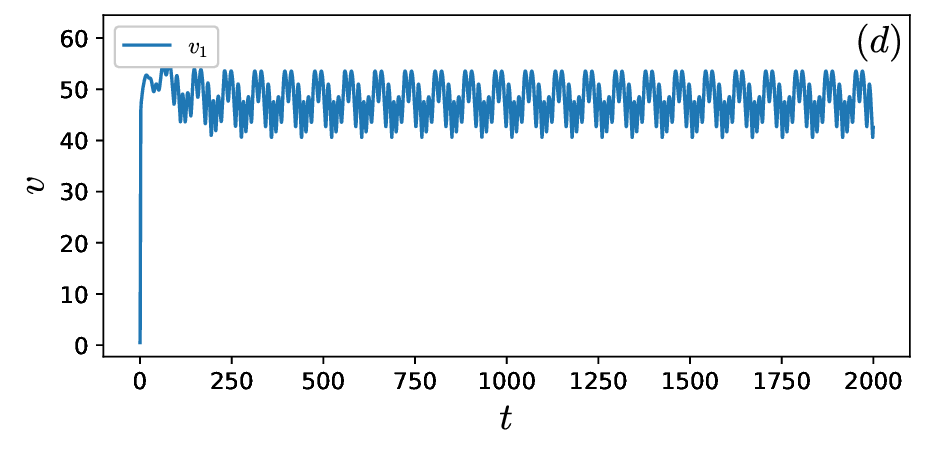}
\includegraphics[width=0.495\textwidth]{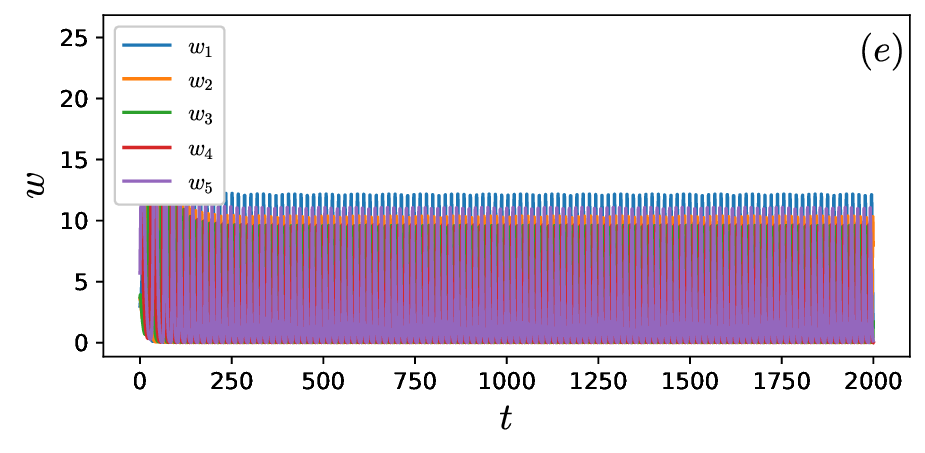}
\includegraphics[width=0.495\textwidth]{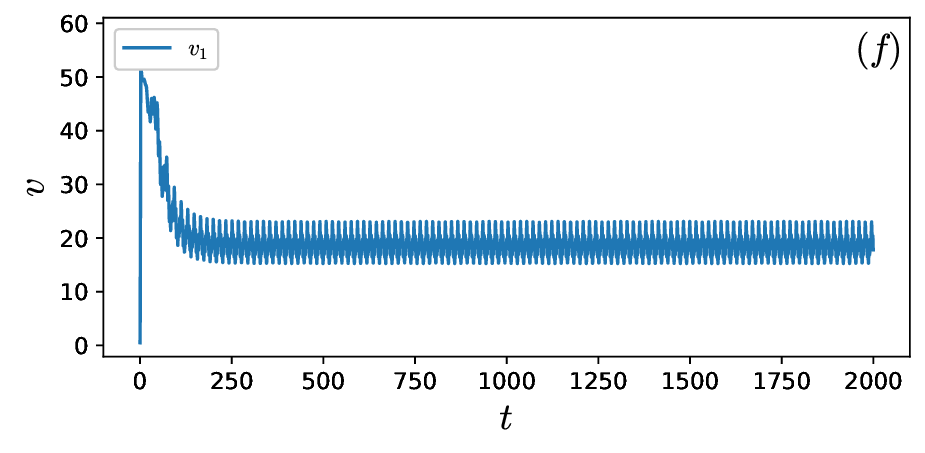}
\includegraphics[width=0.495\textwidth]{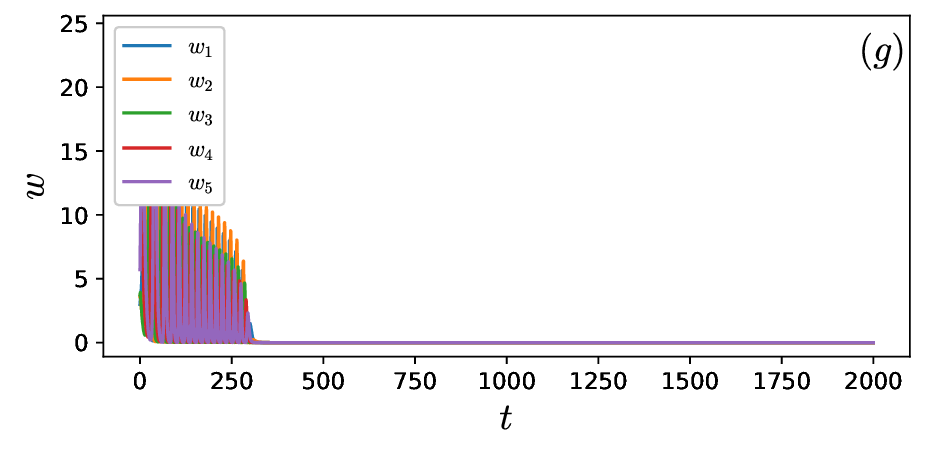}
\includegraphics[width=0.495\textwidth]{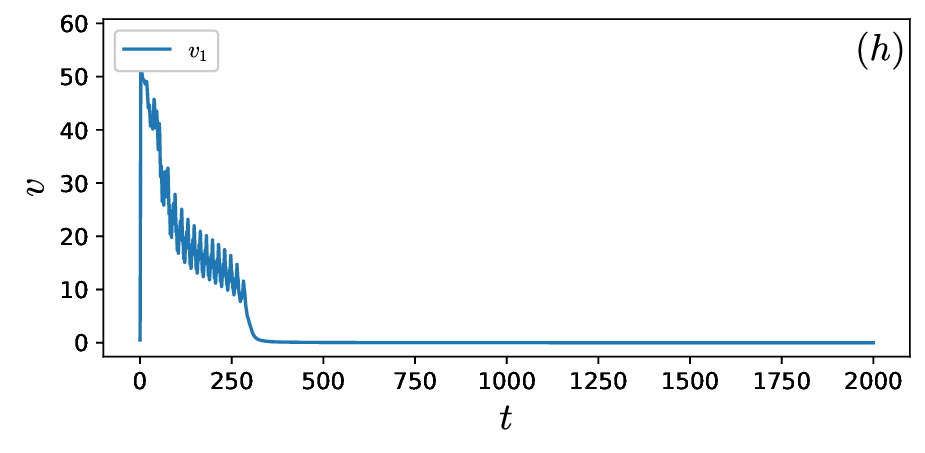}
\caption{Coevolutionary dynamics of replicator--virus system in Example \ref{ex4:4}. In this case the virus is given a strong evolutionary advantage, which eventually, through intermediate destroying the replicator, leads to a full system collapse. $(a),(b)$ --- the dynamics of replicator and virus at evolutionary time $\tau=0$. $(c), (d)$ --- the dynamics of replicator and virus after 300 evolutionary steps. $(e), (f)$ --- the dynamics of replicator and virus after 450 evolutionary steps. $(g), (h)$ --- the dynamics of replicator and virus at 451 evolutionary step (the algorithm ceases to work).}\label{fig4:12}
\end{figure}

\end{example}
\newpage

\section{Conclusions}

In this paper we introduce a novel mathematical framework to describe possible ecological and evolutionary consequences of consumer--resource dynamics. The framework is based on several key assumptions, which can be considered as minimal biologically realistic ingredients that take into account the interactions within the consumer and resource (sub)populations as well as the interactions between them.

The starting point of our framework is the classical replicator equation, which we modify to incorporate a possibility of population extinction --- a natural outcome in many consumer--resource interactions. We call such replicator equations, which are defined on the non-negative cone in $\R^n$, the open replicator systems. In our framework, these open systems represent the resource. The consumer counterpart is assumed to be fully dependent on the presence of the resource, and this motivates the specific coupling of the open replicator equation with another subsystem of ordinary differential equations, which describes the consumer reproduction process. Naturally, in the absence of resource, the consumer dies out. At this point we assume that the interaction process occurs on longer times, such that it is reasonable to consider not only ecological, but also evolutionary changes. These changes are described by the possible changes of the elements of the matrices that represent the network of interaction of replicator macromolecules and the intensities of transforming the resource concentrations into the replication rates of the consumer. From an abstract mathematical point of view we are faced with a situation that resembles a differential game in which two agents are choosing their decisions on the current behavior of the opponent and their own objective function. It is hardly feasible, even numerically, to find a solution of such mathematical problem, and therefore we replace it with a somewhat heuristic approach, in which the evolution of populations is divided into smaller steps. The key idea is based on the previously studied principle of evolutionary adaptation \cite{bratus2024food,bratus2018evolution,drozhzhin2021fitness}, the difference is that now we switch the objective function as each step, thus allowing for the consumer and the resource to evolve according to their own evolutionary goals. In short, we call this sequential process ``defence--attack.''

To summarize, we present a mathematical model of consumer--resource type with clear opposite evolutionary objectives. The model is certainly beyond any exhaustive analytical investigation, and yet is simple enough for straightforward numerical implementation. Our key finding, which is based on numerical examples, is that such basic system, with the two major parts having opposite optimization criteria, is capable of producing on its own rather counterintuitive evolutionary scenarios. Namely, together with somewhat expected possibilities of either consumer or replicator ``winning'' the evolutionary battle (Examples \ref{ex4:3} and \ref{ex4:4}), it is possible to observe quasistationary situations (Examples \ref{ex4:1} and \ref{ex4:2}), in which both the consumer and the resource coexist and steadily increase their respective fitnesses with time, as if their mutual interactions aim to support each other whereas in reality their coupling is of host--parasite type. We think that such outcomes are the consequence of the more complicated than usual modeling structure of both intra- and inter- (sup)species interactions (interactions between interactions \cite{moller2008interactions}), which tune the system in a way so that syngeneic actions of evolutionary forces, which separately aim to defeat the opponent, favor unexpected coexistence.


\section{Appendix}\label{ap:1}
\begin{proof}[Proof of Proposition \ref{pr:1}]Let $\IP{\cdot}{\cdot}$ be the standard inner product in $\R^n$. Summing all the equations in \eqref{pr:1} we find that
$$
\dot S=\phi(S)\IP{A\bs w}{\bs w}-\IP{\bs d}{\bs w},\quad S(0)=S_0>0.
$$
Since there always exists constant $K>0$ such that $\IP{A\bs w}{\bs w}\leq K \IP{\bs w}{\bs w}$, and by elementary inequality $\IP{\bs w}{\bs w}\leq S^2$ we find that
$$
\dot S\leq \phi(S)S^2-d_{\min} S.
$$
Since $\sup_{S\geq 0}\phi(S)S^2=C>0$ we conclude, by comparison theorem, that $S(t)\leq \max\{S_0,C d_{\min}^{-1}\}$, hence our solutions, taking into account the invariance of $\R^n_+$, are unique, defined for all $t>0$, and bounded from infinity.
\end{proof}

\begin{proof}[Proof of Proposition \ref{pr:2}]Since by assumption $S(t)>A>0$ for all $t>0$ we can rescale the time in \eqref{eq2:1} to obtain orbitally equivalent system
$$
\dot w_i=w_i\Bigl((\bs{A w})_i-d_i e^{\gamma S}\Bigr).
$$
Dividing by $w_i$, integrating from $0$ to $T$, and dividing by $T$ we find
$$
\frac{\log w_i(T)-\log w_i(0)}{T}=\sum_{j=1}a_{ij}\frac{1}{T}\int_0^T w_i(t)\D t-d_i\frac{1}{T}\int_0^T e^{-\gamma S(t)}\D t.
$$
Since all $w_i(t)$ are bounded there must exist accumulation points
$$
u_i=\lim_{T_k\to\infty}\frac{1}{T_k}\int_0^{T_k} w_i(t)\D t,\quad \beta=\lim_{T_k\to\infty}\frac{1}{T_k}\int_0^{T_k} e^{-\gamma S(t)}\D t,
$$
which satisfy the system
$$
u_i=(\bs{A}^{-1}\bs d)_i\beta=(\bs{A}^{-1}\bs d)_ie^{\gamma \sum_{i=1}^nu_i},
$$
which is the same as \eqref{eq2:2}.
\end{proof}

\begin{proof}[Idea of the proof of Proposition \ref{pr:3}. See \cite{pavlovich2012studying} for more details]Under the given assumptions condition \eqref{eq2:1b} takes the form
$$
(\gamma e)^{-1}<dn
$$
and the following equalities hold
$$
d=\hat{w}_{i-1}^k\exp(-\gamma S(\bs{\hat{w}}^k)),\quad i=1,\ldots,n,\quad k=1,2.
$$
By explicit calculations, the Jacobi matrix of \eqref{eq2:1} with \eqref{eq2:1c} evaluated at $\bs{\hat{w}}^{1,2}$ is the circulant matrix with the first row
$$
\left[-\gamma d^2 e^{\gamma S(\bs{\hat{w}}^k)},\ldots,-\gamma d^2 e^{\gamma S(\bs{\hat{w}}^k)}, d -\gamma d^2 e^{\gamma S(\bs{\hat{w}}^k)}\right],\quad k=1,2.
$$
The first eigenvalue $\lambda_1$ of such circulant matrix is (e.g., \cite{hofbauer1998ega})
$$
\lambda_1=d-nd^2\gamma e^{\gamma S(\bs{\hat{w}}^{k})}.
$$
If $k=1$, i.e., if $S(\bs{\hat{w}}^1)<\gamma^{-1}$ and $nd<(\gamma e)^{-1}$, then $\lambda_1>-nd^2\gamma e+d>0$, hence $\bs{\hat{w}}^1$ is always unstable. Analogous calculations for $\bs{\hat{w}}^2$ show the asymptotic stability for $n=3,4$ and instability for $n\geq 0$.
\end{proof}

\begin{proof}[Proof of Proposition \ref{pr4:1}]Let $\delta B(\tau)=0$ and $\bs R=\bs A-\alpha \bs C$. From \eqref{eq3:4} it follows that
$$
\delta\psi(\bs{\hat w})\bs R \bs{\hat{w}}+\psi(\bs{\hat w})\delta\bs A\bs{\hat{w}}+\psi(\bs{\hat w})\bs R\delta\bs{\hat w}=0.
$$
Since $\delta \psi(\bs{\hat w})=-\gamma \psi(\bs{\hat w})\delta S(\bs{\hat w})$ then
$$
\delta \bs{\hat w}=\gamma \delta S(\bs{\hat w})\bs{\hat w}-\bs R^{-1}\delta \bs A\bs{\hat w}.
$$
Therefore,
$$
\delta S(\bs{\hat w})=\IP{\delta \bs{\hat w}}{\bs 1}=\gamma \delta S(\bs{\hat w})\IP{\bs{\hat w}}{\bs 1}-\IP{\bs R^{-1}\delta \bs A\bs{\hat w}}{\bs 1},
$$
and hence
$$
\delta S(\bs{\hat w})=\frac{\IP{\bs R^{-1}\delta \bs A\bs{\hat w}}{\bs 1}}{\gamma S(\bs{\hat w})-1}\,.
$$
From equality \eqref{eq3:6} it follows that
$$
\bar{f}(\tau)=\exp(\gamma S(\bs{\hat w})).
$$
Therefore,
$$
\delta \bar{f}_A(\tau)=\psi^{-1}(\bs{\hat w})\gamma \delta S(\bs{\hat w})=\frac{\gamma}{\psi(\bs{\hat w})}\frac{\IP{\bs R^{-1}\delta \bs A\bs{\hat w}}{\bs 1}}{\gamma S(\bs{\hat w})-1}\,
$$
as claimed.

Analogous reasoning in the case $\delta\bs A(\tau)=0$ leads to \eqref{eq4:3}.
\end{proof}

\paragraph{Acknowledgements:} A.S.B. is funded by the Russian Science Foundation, grant number 23-11-
00116. A.S.B. and S.V.D are supported by the Moscow Center of Fundamental and Applied Mathematics of Lomonosov Moscow
State University under agreement No. 075-15-2025-345.


\end{document}